\documentclass[11pt,onecolumn,doublespace,draftclsnofoot]{IEEEtran}

\usepackage{epsfig,amsmath,graphicx,subfigure}
\usepackage{times,rawfonts}
\usepackage{pdfsync}
\usepackage{hyperref}
\IEEEoverridecommandlockouts

\newtheorem{lemma}{Lemma}%[section]
\newtheorem{theorem}{Theorem}%[section]
\title{Secret Sharing LDPC Codes for the BPSK-constrained Gaussian
  Wiretap Channel}

\author{Chan Wong Wong, Tan F. Wong, and John M. Shea\thanks{
    This work was supported by the National Science Foundation under
    grant number CNS-0626863.}}

\begin{document}
\maketitle %\setcounter{page}{1}

\begin{abstract}
  The problem of secret sharing over the Gaussian wiretap channel is
  considered. A source and a destination intend to share secret
  information over a Gaussian channel in the presence of a wiretapper
  who observes the transmission through another Gaussian channel. Two
  constraints are imposed on the source-to-destination channel;
  namely, the source can transmit only binary phase shift keyed (BPSK)
  symbols, and symbol-by-symbol hard-decision quantization is applied
  to the received symbols of the destination. An error-free public
  channel is also available for the source and destination to exchange
  messages in order to help the secret sharing process. The wiretapper
  can perfectly observe all messages in the public channel.  It is
  shown that a secret sharing scheme that employs a random ensemble of
  regular low density parity check (LDPC) codes can achieve the key
  capacity of the BPSK-constrained Gaussian wiretap channel
  asymptotically with increasing block length. To accommodate
  practical constraints of finite block length and limited decoding
  complexity, fixed irregular LDPC codes are also designed to replace
  the regular LDPC code ensemble in the proposed secret sharing
  scheme.
\end{abstract}

\section{Introduction}\label{sec:intro}
Physical-layer security schemes exploit channel characteristics, such
as noise and fading, to allow a group of nodes to share information in
such a way that other unintended receivers (called eavesdroppers or
wiretappers) cannot recover that secret information.  Physical-layer
security has often been studied in the context of the {\it wiretap
  channel}, which was first introduced by Wyner~\cite{Wyner1975} and
later refined by Csisz\'{a}r and K\"{o}rner~\cite{Csiszar1978}. In the
wiretap channel, a source tries to send secret information to a
destination at the presence of a wiretapper. When the
source-to-wiretapper channel is a degraded version of the
source-to-destination channel, Wyner~\cite{Wyner1975} showed that the
source can transmit a message at a positive (secrecy) rate to the
destination by taking advantage of the less ``noisy'' channel to the
destination.  The degradedness condition was removed
in~\cite{Csiszar1978}, which showed that a positive secrecy rate is
possible for the case where the source-to-destination channel is
``more capable'' than the source-to-wiretapper channel. Generalization
of Wyner's work to the Gaussian wiretap channel was considered
in~\cite{Leung1978}.

In Wyner's original paper, a code design based on group codes was
described for the wiretap channel.  In~\cite{Ozarow1984}, a code
design based on coset codes was suggested for the type-II (the
destination channel is error free) binary erasure wiretap channel.
Recently, the authors of~\cite{Thangaraj2007} constructed low density
parity check (LDPC) based wiretap codes for binary erasure channel
(BEC) and binary symmetric channel
(BSC). Reference~\cite{LiuRuoheng07} considered the design of secure
nested codes for type-II wiretap channels.  More recently,
References~\cite{Mahdavifar2010} and \cite{Koyluoglu2010} concurrently
established the result that polar codes~\cite{Arikan2009} can achieve
the secrecy capacity of the degraded binary-input symmetric-output
(BISO) wiretap channels. Note that all these designs are for codes
with asymptotically large block lengths.

In some scenarios, it is sufficient for two nodes to agree upon a
common secret (a key), instead of having to send secret information
from a source to a destination.  Under this relaxed criterion, it is
shown in~\cite{Maurer1993} that, with the use of a feedback channel, a
positive key rate is achievable when the destination and wiretapper
channels are two conditionally independent (given the source input
symbols) memoryless binary channels, even if the destination channel
is not more capable than the wiretapper's channel.  This notion of
secret sharing is formalized in~\cite{Ahlswede1993} based on the
concept of \emph{common randomness} between the source and
destination.  A three-phase process of achieving secret sharing over a
wiretap channel with an additional public channel between the source
and destination is suggested in~\cite{Maurer1993}.  The three phases
are respectively advantage distillation, information reconciliation,
and privacy amplification.  Advantage distillation aims to provide the
destination an advantage over the wiretapper.  Information
reconciliation aims at generating an identical random sequence at both
the source and destination.  Privacy amplification is the step that
extracts a secret key from the identical random sequence agreed upon
by the source and destination.

Information reconciliation is probably the most studied and most
essential part of any secret sharing scheme.
%It falls into the category of secrecy extraction from correlated
%sources and has close connections to the problem of source coding with
%side information.
Perhaps the most well known practical application of reconciliation
protocols is quantum cryptography, where nonorthogonal states of a
quantum system provide two terminals with observations of correlated
randomness which are at least partially secret from a potential
eavesdropper.  Many
works~\cite{Brassard94}--\nocite{Nguyen04,VanAssche2004,Muramatsu2006,Ye2006,Bloch2008}\cite{Elkouss2009}
have been devoted to the study of reconciliation for both discrete and
continuous random variables in quantum key distribution schemes.  For
the case of discrete random variables, Cascade is an iterative
reconciliation protocol first proposed by Brassard and Salvail
in~\cite{Brassard94}.
%Despite being highly interactive, Cascade is the most widely used
%reconciliation protocol in practical QKD setups because of its
%simplicity and reasonably efficiency.  variations around the principle
%of interactive reconciliation used in Cascade have then been proposed
%to limit the interactivity.
Recently, BSC-optimized LDPC codes have been employed
in~\cite{Elkouss2009} to reduce the interactivity and improve the
efficiency of Cascade.  On the other hand, the work on slice error
correction~\cite{VanAssche2004}, which converts continuous variables
into binary strings and makes use of interactive error correcting
codes, is the first reconciliation protocol for continuous random
variables.  Modern coding techniques like turbo codes~\cite{Nguyen04}
and LDPC codes~\cite{Ye2006,Bloch2008} have also been directly applied
within information reconciliation protocols for continuous random
variables.

Another application of reconciliation protocols is secret key
agreement over wireless channels.  An LDPC code-based method of
extracting secrecy from jointly Gaussian random sources generated by a
Rayleigh fading model has been studied in~\cite{Ye2006}.
In~\cite{Bloch2008}, multilevel coding/multistage decoding-like
reconciliation with LDPC codes has been proposed for the quasi-static
Rayleigh fading wiretap channel.  In~\cite{Klinc2009}, punctured LDPC
codes were employed in a coding scheme for the Gaussian wiretap
channel to reduce the security gap, which expresses the quality
difference between the destination channel and wiretapper channel
required to achieve a sufficient level of security.  The main idea of
this scheme is to hide the information bits from the wiretapper by
means of puncturing. In~\cite{Baldi2010}, further reductions in the
security gap are achieved using a reconciliation scheme based on
non-systematic LDPC codes along with scrambling of the information
bits prior to encoding.

In this paper, we consider the problem of secret sharing over the
Gaussian wiretap channel with the constraints of binary phase-shift
keyed (BPSK) source symbols and symbol-by-symbol hard-decision
quantization at the destination. Our main goal is to develop a coding
structure based on which practical ``close-to-capacity'' secret
sharing (key agreement) codes can be constructed. Finite block length
and moderate encoder/decoder complexity are the two main practical
constraints that we consider when designing these codes.  The secrecy
performance of our designs will be measured by the rate of secret
information shared between the source and destination (which will be
referred to as the \emph{key rate}) as well as the rate of information
that is leaked to the wiretapper through all its observations of the
wiretap and public channels (which will be referred to as the
\emph{leakage rate}).

To rigorously gauge the secrecy performance of our code designs, we
introduce the notion of relaxed key capacity in
Section~\ref{sec:relax_keycap}. The relaxed key capacity is the
maximum key rate that can be achieved over the wiretap channel
provided that the leakage rate is bounded below a fixed value. In
Section~\ref{sec:ch_mod}, we calculate the relaxed key capacities over
the BPSK source-constrained Gaussian wiretap channel with and without
the constraint of hard-decision quantization at the destination. In
Section~\ref{sec:ssldpc}, we present a secret sharing scheme employing
an ensemble of regular LDPC codes for the BPSK-constrained Gaussian
wiretap channel with hard-decision quantization at the destination.
We prove that the proposed scheme achieves the relaxed key capacity
with asymptotically large block length.  We note that a similar
LDPC-based key agreement scheme employing observations of correlated
discrete stationary sources at the source, destination, and wiretapper
was studied in~\cite{Muramatsu2006}. A more detailed comparison
between our scheme and the one proposed in~\cite{Muramatsu2006} is
also provided in Section~\ref{sec:ssldpc}.  The asymptotic result in
Section~\ref{sec:ssldpc} provides us a reasonable theoretical
justification to design practical secret sharing schemes based on the
proposed coding structure.  We propose in Section~\ref{sec:wz} to
replace the regular LDPC code ensemble in Section~\ref{sec:ssldpc} by
fixed LDPC codes that are more amenable to practical
implementation. In the same section, we describe a code search
algorithm based on density evolution analysis to obtain good irregular
LDPC codes for the proposed secret sharing scheme. We also compare the
secrecy performance achieved by these irregular LDPC codes,
BSC-optimized irregular LDPC codes, and some standard regular LDPC
codes against the relaxed key capacity calculated in
Section~\ref{sec:ch_mod}.  Finally, conclusions are drawn in
Section~\ref{sec:con}.

\section{Secret Sharing and Relaxed Key
  Capacity}\label{sec:relax_keycap}

We start by reviewing the framework of secret sharing proposed in
\cite{Ahlswede1993}.  The objective of secret sharing is for the
source and destination to share secret information, which is obscure
to the wiretapper, by exploiting \emph{common
  randomness}~\cite{Ahlswede1993} available to them through the
wiretap channel.  Here, we consider the wiretap channel to be
memoryless and specified by the conditional probability density
function (pdf) $p_{Y,Z|X}(y,z|x)$. When the symbol $X$ is sent by the
source, $Y$ and $Z$ denote the corresponding symbols observed by the
destination and wiretapper, respectively. In addition, we restrict
ourselves to cases in which $Y$ and $Z$ are conditionally independent
given $X$, i.e., $p_{Y,Z|X}(y,z|x) = p_{Y|X}(y|x) p_{Z|X}(z|x)$.  This
restriction is satisfied by the Gaussian wiretap channel considered in
Section~\ref{sec:ch_mod} and some other wireless wiretap channels
\cite{Wong2009}.  For convenience, we will refer to the wiretap
channel by the triple $(X,Y,Z)$.  In addition to the wiretap channel,
there is an interactive, authenticated, public channel with unlimited
capacity between the source and destination. The source and
destination can communicate via the public channel without any power
or rate restriction.  The wiretapper can perfectly observe all
communications over the public channel but cannot tamper with the
transmitted messages.

The aforementioned common randomness is to be extracted by a proper
combination of transmission from the source to the destination through
the wiretap channel $(X,Y,Z)$ and information exchanges between them
over the public channel. To this end, we consider the class of
permissible secret sharing strategies suggested
in~\cite{Ahlswede1993}. Consider $t$ time instants labeled by
$1,2,\ldots,t$, respectively. The wiretap channel is used $n$ times
during these $t$ time instants at $i_1 < i_2 < \cdots < i_n$.  Set
$i_{n+1}=t$. The public channel is used during the other ($t-n$) time
instants. Before the secret sharing process starts, the source and
destination generate, respectively, independent random variables $M_X$
and $M_Y$.  Then a permissible strategy proceeds as
follows:\footnote{Throughout the paper, $A^i$ stands for the sequence
  of symbols $A_1, A_2, \ldots, A_i$, and $A^0$ is null.}
\begin{itemize}
\item At time instant $0<i<i_1$, the source sends message
  $\Phi_i=\Phi_i(M_X,\Psi^{i-1})$ to the destination, and the
  destination sends message $\Psi_i=\Psi_i(M_Y,\Phi^{i-1})$ to the
  source. Both transmissions are carried over the public channel.
\item At time instant $i=i_j$ for $j=1,2,\ldots,n$, the source sends
  the symbol $X_j=X_j(M_X,\Psi^{i_j-1})$ to the wiretap
  channel. The destination and wiretapper observe the corresponding
  symbols $Y_j$ and $Z_j$. There is no message exchange via the public
  channel, i.e., $\Phi_i$ and $\Psi_i$ are both null.
\item At time instant $i_j < i < i_{j+1}$ for $j=1,2,\ldots,n$, the
  source sends message $\Phi_i=\Phi_i(M_X,\Psi^{i-1})$ to the
  destination, and the destination sends message
  $\Psi_i=\Psi_i(M_Y,Y^{j},\Phi^{i-1})$ to the source. Both
  transmissions are carried over the public channel.
\end{itemize}
At the end of the $t$ time instants, the source generates its secret
key $K=K(M_X,\Psi^t)$, and the destination generates its secret key
$L=L(M_Y,Y^n,\Phi^t)$, where $K$ and $L$ take values from the same
finite set $\mathcal{K}$.

Slightly extending the achievable key rate definition in
\cite{Ahlswede1993}, for $R_l \geq 0$, we call $(R,R_l)$ an
\emph{achievable key-leakage rate pair} through the wiretap channel
$(X,Y,Z)$ if for every $\varepsilon>0$, there exists a permissible
secret sharing strategy of the form described above such that
\begin{enumerate}
\item $\Pr\{K\neq L\} < \varepsilon$,
\item $\frac{1}{n} I(K;\Phi^t, \Psi^t) < \varepsilon$,
\item $\frac{1}{n} I(K;Z^n| \Phi^t, \Psi^t) < R_l+\varepsilon$,
\item $\frac{1}{n} H(K) > R - \varepsilon$, and
\item $\frac{1}{n} \log_2 |\mathcal{K}| < \frac{1}{n} H(K) +
  \varepsilon$,
\end{enumerate}
for sufficiently large $n$. Condition 2 restricts that the public
messages (the messages conveyed through the public channel) contain a
negligible rate of information about the key, while Condition 3 limits
to $R_l$ the rate of key information that the wiretapper can extract
from its own channel observations and the public messages.  Note that
Condition 3 is trivially satisfied if $R_l \geq \frac{1}{n} \log_2
|\mathcal{K}|$.  We also note that Conditions 2 and 3 combine to
essentially give the original condition $\frac{1}{n}
I(K;Z^n,\Phi^t,\Psi^t) < \varepsilon$ of the achievable key rate
definition in~\cite{Ahlswede1993} when $R_l=0$\footnote{When $R_l>0$,
  if the combined condition $\frac{1}{n} I(K;Z^n,\Phi^t,\Psi^t) <
  R_l+\varepsilon$ is employed instead of Conditions~2 and 3, then it
  is easy to see that if $(R,R_l)$ is an achievable key-leakage rate
  pair, $(R+r,R_l+r)$ is also achievable, for any $r\geq 0$, by simply
  transmitting the additional key information (of rate $r$) through
  the public channel. Separating the two conditions as suggested
  avoids such artificial consequence of the combined condition.}.  For
the cases in which the alphabet of $X$ is not finite, we also impose
the following power constraint to the symbol sequence $X^n$ sent out
by the source:
\begin{equation} \label{eqn:powerconstr}
\frac{1}{n} \sum_{j=1}^{n} |X_j|^2 \leq P
\end{equation}
with probability one (w.p.1) for sufficiently large $n$.  We note that
the idea of key-leakage rate pair is similar to that of the
secrecy-equivocation rate pair originally defined in~\cite{Wyner1975}.

The \emph{$R_l$-relaxed key capacity} is defined as the maximum value
of $R$ such that $(R,R_l)$ is an achievable key-leakage rate pair. The
main reason for us to introduce the notion of relaxed key capacity is
to employ it as a gauge to measure the performance of practical codes
that will be presented in \autoref{sec:wz}. Since these codes have
finite block lengths and are to be decoded by the belief propagation
(BP) algorithm, they do not achieve zero leakage rate. Thus using the
relaxed key capacity provides a more suitable comparison than using
the original ``straight'' key capacity in~\cite{Ahlswede1993}. Also,
since these practical codes do not give zero leakage rate, their use
could be considered as an information-reconciliation step. The secrecy
performance could be further improved by additional privacy
amplification.
 
For wiretap channels that satisfy the aforementioned conditional
independence requirement, we have the following result, whose proof is
sketched in Appendix~\ref{app:pf_keycapgen}:
\begin{theorem}\label{thm:keycapgen}
  The $R_l$-relaxed key capacity of the memoryless wiretap channel
  $(X,Y,Z)$ with conditional pdf $p(y,z|x)=p(y|x)p(z|x)$ is given by
\[
%\label{eqn:relaxkeycap} 
C_K(R_l) = \max_{X: E[|X|^2] \leq P} \left[
  \min \{I(X;Y) - I(Y;Z) + R_l, I(X;Y)\} \right].
\]
\end{theorem}
%\begin{proof}
%  Both the converse and achievability proofs of \cite[Theorem
%  1]{Wong2009}, which corresponds to the case when $R_l=0$, can be
%  easily extended to accommodate Conditions 2 and 3 in the above
%  definition of achievable key-leakage rate pair. The details are
%  omitted due to space limitation.
%\end{proof}
We employ this result to calculate the relaxed key capacities of the
BPSK-constrained Gaussian wiretap channel in the next section.

\section{BPSK-constrained Gaussian Wiretap Channel}\label{sec:ch_mod}

Hereafter, we focus on the Gaussian wiretap channel, in which the
source-to-destination channel and source-to-wiretapper channel are
both additive white Gaussian noise (AWGN) channels. We restrict the
source to transmit only BPSK symbols. More specifically, let $X_i \in
\{\pm 1\}$ be the $i$th transmit symbol from the source\footnote{In
  later sections, whenever appropriate, we implicitly employ the
  mapping $+1 \rightarrow 0$ and $-1 \rightarrow 1$, where $0$ and $1$
  are the two usual elements in GF(2).}, and let $Y_i$ and $Z_i$ be
the corresponding received symbols at the destination and wiretapper,
respectively. The Gaussian wiretap channel can then be modeled as
\begin{equation}\label{eqn:channel}
\begin{split}
Y_i &= \beta X_i + N_i \\
Z_i &= \alpha \beta X_i + \tilde N_i,
\end{split}
\end{equation}
where $N_i$ and $\tilde N_i$ are i.i.d. zero-mean Gaussian random
variables of variance $\sigma^2$. Note that $\beta$ is the gain of the
BPSK symbols transmitted by the source. By the source power
constraint~\eqref{eqn:powerconstr}, we have $\beta^2 \leq P$. Also,
$\alpha$ is a positive constant that models the gain advantage of the
wiretapper over the destination. Let the normalized gain $\tilde \beta
= \beta / \sigma$. Then, the received signal-to-noise ratios (SNRs) at
the destination and wiretapper are $\tilde \beta^2$ and $\alpha^2
\tilde \beta^2$, respectively. Clearly, the Gaussian wiretap channel
satisfies the memoryless and conditional independent properties
required in Theorem~\ref{thm:keycapgen}.  Specializing
Theorem~\ref{thm:keycapgen} to the BPSK-constrained Gaussian wiretap
channel, it is not hard to show\footnote{The proofs of
  \eqref{eqn:relaxkeycap_bpsk} and \eqref{eqn:relaxkeycap_bpsk_constr}
  can be easily, though rather tediously, established by checking the
  concavity and symmetry of $I(X;Y)-I(Y;Z)$ as a function of the
  binary source distribution in the respective cases.} that the
$R_l$-relaxed key capacity is given by
\begin{eqnarray}
  C_b(R_l) &=& 
  \max_{0 \leq \tilde \beta \leq \sqrt{\frac{P}{\sigma^2}}} \Bigg\{ \min\bigg\{
  \frac{1}{2\pi}\int_{0}^{\infty}\!\!\!\!\int_{0}^{\infty} H_2 \left(
    \frac{1+e^{-2\tilde \beta y}\cdot e^{-2\alpha\tilde \beta
        z}}{[1+e^{-2\tilde \beta
        y}][1+e^{-2\alpha\tilde \beta z}]}\right) \left[1+e^{-2\tilde \beta y}\right]
  \left[1+e^{-2\alpha\tilde \beta z}\right] \nonumber \\
  & & ~~~~~~~~~~~~~~~~~~~~~~~~~~~~~~~~~~\cdot
  \exp\left[-\frac{(y-\tilde \beta)^2}{2}-\frac{(z-\alpha\tilde
      \beta)^2}{2}\right]dydz +R_l,1 \bigg\}
  \nonumber \\
  & & ~~~~~~~~~~~~~~~~
  - \frac{1}{\sqrt{2\pi}}
  \int_{0}^{\infty} H_2 \left(\frac{1}{1+e^{-2\tilde \beta y}}\right)
  \left(1+e^{-2\tilde \beta y}\right)
  \exp\left[-\frac{(y-\tilde \beta)^2}{2}\right] dy \Bigg\}
\label{eqn:relaxkeycap_bpsk}
\end{eqnarray}
where $H_2(p) = -p\log_2 p -(1-p) \log_2 (1-p)$ is the binary entropy
function. We note that $C_b(R_l)$ is achieved when $X_i$ is
equiprobable; but it is not necessarily achieved by transmitting at
the maximum allowable power $P$.

The achievability proof of Theorem~\ref{thm:keycapgen}
(cf. Appendix~\ref{app:pf_keycapgen}) employs random Wyner-Ziv coding,
in which the received symbols at the destination need to be quantized
due to the fact that the channel alphabet at the destination in the
Gaussian wiretap channel is continuously distributed. In this paper,
we consider a simple symbol-by-symbol hard-decision quantization
scheme in which the $i$th quantized destination symbol $\tilde Y_i =
\mbox{sgn}(Y_i)$, where $\mbox{sgn}$ is the signum function. Note that
this quantization is suboptimal and leads to a loss in key
capacity. We quantify this loss by applying
Theorem~\ref{thm:keycapgen} to the BPSK-constrained Gaussian wiretap
channel with hard-decision quantization at the destination to
calculate the relaxed-$R_l$ key capacity $C_{bq}(R_l)$. Using the
standard notation
$Q(x)=\int_{x}^{\infty}\frac{e^{-u^2/2}}{\sqrt{2\pi}}du$, it is not
hard to establish\footnotemark[4] that
\begin{equation}
C_{bq}(R_l) = 
\max_{0 \leq \tilde \beta \leq \sqrt{\frac{P}{\sigma^2}}} \left[
  \min\{C_s(\tilde\beta) - C_w(\tilde\beta) +R_l, C_s(\tilde\beta)\} \right],
\label{eqn:relaxkeycap_bpsk_constr}
\end{equation}
where
\begin{eqnarray}
C_s(\tilde\beta) &=& 1 -  H_2(Q(\tilde \beta)) 
\label{eqn:Cs} \\
C_w(\tilde\beta) &=&
1 - \frac{1}{\sqrt{2\pi}} \int_0^{\infty} H_2\left(
\frac{Q(\tilde \beta)+[1-Q(\tilde \beta)]e^{-2\alpha \tilde \beta
z}}{1+e^{-2 \alpha \tilde \beta z}}\right)[1+e^{-2 \alpha \tilde
\beta z}] e^{-\frac{(z-\alpha\tilde \beta)^2}{2}}dz.
\label{eqn:Cw}
\end{eqnarray}
are respectively the capacities of the quantized-destination-to-source
and quantized-destination-to-wiretapper channels at the normalized
gain $\tilde\beta$. Like before, $C_{bq}(R_l)$ is achieved when $X_i$
is equiprobable; but it is not necessarily achieved by transmitting at
the maximum allowable power $P$.  To visualize the loss in key
capacity, \autoref{fig:ck_p} shows $C_b(R_l)$ and $C_{bq}(R_l)$ versus
the maximum allowable SNR ($P/\sigma^2$) for different values of
$R_l$.  We can see that the loss in key capacity due to the
hard-decision quantization is no more than $0.07$ bits per (wiretap)
channel use for the cases shown.

\section{Secret sharing scheme employing regular LDPC code
  ensembles}\label{sec:ssldpc}

The achievability proof of Theorem~\ref{thm:keycapgen} in
Appendix~\ref{app:pf_keycapgen} employs a secret sharing scheme with
random Wyner-Ziv coding.  For the BPSK-constrained Gaussian wiretap
channel with destination hard-decision quantization, we show in this
section that a secret sharing scheme that employs a properly
constructed ensemble of regular LDPC codes can also asymptotically
achieve the $R_l$-relaxed key capacity. We will design practical
secret sharing schemes for the BPSK-constrained Gaussian wiretap
channel in Section~\ref{sec:wz} based on the LDPC coding structure
proposed here.

To start describing the proposed secret sharing scheme, let us
consider an $(n,l)$ binary linear block code $\mathcal{C}$ with $2^l$
distinct codewords of length $n$ and an $(l-k)$-dimensional subspace
$\mathcal{W}$ in $\mathcal{C}$. The pair $(\mathcal{C},\mathcal{W})$
defines what we call an $(n,l,k)$ \emph{secret sharing binary linear
  block code}. Given any such $(\mathcal{C},\mathcal{W})$ pair, let
$\mathcal{K}$ be the quotient of $\mathcal{C}$ by $\mathcal{W}$. Then
$\mathcal{K}$ is a linear space of $2^k$ distinct cosets of the form
$\hat x^n+\mathcal{W}$, where $\hat x^n \in \mathcal{C}$.  We will use
the coset index in $\mathcal{K}$ as the secret key. We will see later
that the ordering of the cosets in $\mathcal{K}$ is immaterial. The
ratios $R_c=\frac{l}{n}$ and $R_k=\frac{k}{n}$ will be referred to as
the \emph{code rate} and \emph{key rate} of the $(n,l,k)$
secret sharing binary linear block code, respectively.

Next, we consider the following random ensemble of $(n,l,k)$
secret sharing binary linear block codes:
\begin{itemize}
\item The $(n,l)$ linear block code $\mathcal{C}$ is chosen uniformly
  from the ensemble of $(d_v,d_c)$-regular LDPC codes considered in
  \cite{RichardsonIT01_1}.  That is, we consider that $\mathcal{C}$ is
  chosen uniformly from the set of all bipartite graphs
  \cite{tanner1981ral} with $n$ degree-$d_v$ variable nodes and $n-l$
  degree-$d_c$ check nodes.
\item The subspace $\mathcal{W}$ is chosen uniformly over the set of
  all possible $(l-k)$-dimensional subspaces in $\mathcal{C}$.
\end{itemize}
Note that a realization of the randomly chosen $\mathcal{C}$ may
actually have $2^{l'}$ distinct codewords, where $l'>l$. In such case,
$\mathcal{K}$ will be of dimension $k+l'-l$; so the actual key rate
will be larger than $R_k$. Hence, we can conservatively assume
$\mathcal{C}$ is always an $(n,l)$ linear code with $2^l$ distinct
codewords to simplify the notation below.

Consider the following secret sharing scheme:
\begin{enumerate}
\item {\bf Random source transmission and destination quantization:}
  The source randomly generates a sequence $X^n$ of $n$ i.i.d. equally
  likely BPSK symbols and transmits them consecutively over the
  Gaussian wiretap channel $(X,Y,Z)$. The destination receives the
  sequence $Y^n$ and obtains the quantized sequence $\tilde Y^n$ by
  performing symbol-by-symbol hard-decision quantization on $Y^n$,
  i.e., $\tilde Y_j=\mbox{sgn}(Y_j)$. This quantization effectively
  turns the source-to-destination channel into a BSC, whose cross-over
  probability depends on the SNR of the original source-to-destination
  channel. We note that the wiretapper also observes $Z^n$ through the
  source-to-wiretapper channel.

\item {\bf Syndrome generation through LDPC encoding at destination:}
  The next step is for the destination to feed a compressed version of
  $\tilde Y^n$ back to the source through the public channel so that
  the source can resolve the differences between $X^n$ and $\tilde
  Y^n$.  This is similar to the problem of compressing an equiprobable
  memoryless binary source with side information using LDPC codes
  considered in~\cite{LiverisLDPC02}.  More precisely, the destination
  selects $(\mathcal{C},\mathcal{W})$ randomly from the ensemble of
  secret sharing $(d_v,d_c)$-regular LDPC codes described above.  It
  then generates the syndrome sequence $S^{n-l}=\tilde Y^n H^T$, where
  $H$ is a parity check matrix of $\mathcal{C}$. We note that each
  $S^{n-l}$ uniquely corresponds to a coset $E_S^n +
  \mathcal{C}$. Further, the destination determines which coset in
  $\mathcal{K}$ that $X_0^n=\tilde Y^n + E_S^n \in \mathcal{C}$
  belongs. Denote that coset by $\hat X^n_0 + \mathcal{W}$.  Finally,
  the destination sends $E_S^{n}$, $\mathcal{C}$, and $\mathcal{W}$
  back to the source via the public channel.

\item {\bf Decoding at source:} The source then tries to decode for
  $X_0^n$ from observing $X^n$ and $E_S^{n}$ according to
  $(\mathcal{C},\mathcal{W})$. Treating $X^n + E_S^n$ as a noisy
  version of $X^n_0$, it performs maximum likelihood (ML) decoding to
  obtain a codeword in $\mathcal{C}$ and then determines which coset
  in $\mathcal{K}$ that the decoded codeword belongs. Denote that
  coset by $\hat X^n +\mathcal{W}$.

\item {\bf Key generation at source and destination:} The destination
  sets its key $L$ to be index of $\hat X^n_0+ \mathcal{W}$ in
  $\mathcal{K}$.  Similarly, the source sets its key $K$ to be the
  index of $\hat X^n+ \mathcal{W}$ in $\mathcal{K}$.
\end{enumerate}
It is clear that this secret sharing scheme is permissible. Indeed,
under the notation of \autoref{sec:relax_keycap}, for the proposed
secret sharing scheme, $t=n+1$, $i_j=j$ for $j=1,2,\ldots,n$, $M_X =
X^n$, $M_Y=(\mathcal{C},\mathcal{W})$, and
$\Psi_{n+1}=(E_S^{n},\mathcal{C},\mathcal{W})$ is the only message
sent via the public channel. Hence, we can evaluate the secrecy
performance of the scheme in the context of its achievable key rate
defined in \autoref{sec:relax_keycap} as follows.

First, based on the linearity of LDPC codes, the memoryless nature of
the Gaussian wiretap channel, the chosen distribution of $X^n$, and
the symbol-by-symbol hard decision performed to obtain $\tilde Y^n$ at
the destination, it is easy to check that $H(\tilde Y^n)=n$,
$H(E_S^{n}|\mathcal{C},\mathcal{W})=n-l$,
$H(L|\mathcal{C},\mathcal{W})=k$, and
$I(L;E_S^{n}|\mathcal{C},\mathcal{W})=0$. Then,
\begin{equation*}
0 \leq I(L;E_S^{n},\mathcal{C},\mathcal{W}) =
I(L;\mathcal{C},\mathcal{W}) %+ I(L;S^{n-k}|\mathcal{C},\mathcal{W})
= H(L) - H(L|\mathcal{C},\mathcal{W}) \leq k - k = 0.
\end{equation*}
Hence, $I(L;E_S^{n},\mathcal{C},\mathcal{W}) = 0$,
$I(L;\mathcal{C},\mathcal{W})=0$, and $H(L)=k$. If the decoding
process at the source achieves the ensemble average error probability
$\bar \epsilon_s$, then we have $\Pr\{K \neq L\} \leq \bar \epsilon_s
$. Thus, $H(K|L) \leq 1+k \bar \epsilon_s$ and $H(L|K) \leq 1+k \bar
\epsilon_s$ by Fano's inequality. That in turn implies
\begin{eqnarray*}
  \frac{1}{n}I(K;E_S^{n},\mathcal{C},\mathcal{W})
  &=& \frac{1}{n}\left[I(L;E_S^{n},\mathcal{C},\mathcal{W}) 
    + I(K;E_S^{n},\mathcal{C},\mathcal{W}|L) 
    - I(L;E_S^{n},\mathcal{C},\mathcal{W}|K)\right] \\
  &\leq& 
  \frac{1}{n}I(K;E_S^{n},\mathcal{C},\mathcal{W}|L) \leq \frac{1}{n}H(K|L) 
  \leq R_k\bar \epsilon_s+\frac{1}{n}
\end{eqnarray*}
and
\begin{equation}
  \frac{1}{n}H(K)
  = \frac{1}{n}\left[H(L) + H(K|L) - H(L|K)\right]
  \geq R_k - R_k\bar \epsilon_s - \frac{1}{n}.
\label{eqn:C4}
\end{equation}
Hence, Conditions 2 and 5 in \autoref{sec:relax_keycap} are
satisfied when $n$ is sufficiently large if $\bar \epsilon_s$ can be
made arbitrarily small.  Similarly,
\begin{eqnarray} \label{eqn:IKZS}
\lefteqn{I(K;Z^n,E_S^{n},\mathcal{C},\mathcal{W})}\nonumber \\
  &=& I(L;Z^n,E_S^{n},\mathcal{C},\mathcal{W}) 
  + I(K;Z^n,E_S^{n},\mathcal{C},\mathcal{W}|L) 
  - I(L;Z^n,E_S^{n},\mathcal{C},\mathcal{W}|K) \nonumber \\
  &\leq& 
  I(L;Z^n,E_S^{n},\mathcal{C},\mathcal{W}) 
  + I(K;Z^n,E_S^{n},\mathcal{C},\mathcal{W}|L) \nonumber \\
  &\leq& 
  I(L;Z^n,E_S^{n},\mathcal{C},\mathcal{W}) + H(K|L) \nonumber \\
  &\leq& 
  I(L;Z^n,E_S^{n},\mathcal{C},\mathcal{W}) + k\bar \epsilon_s + 1  \nonumber \\
  &=& I(L;Z^n,E_S^{n}|\mathcal{C},\mathcal{W}) + k\bar \epsilon_s + 1,
\end{eqnarray}
where the last line is due to the fact that
$I(L;\mathcal{C},\mathcal{W}) =0$. Here, 
\begin{eqnarray} \label{eqn:ILZS}
  \lefteqn{I(L;Z^n,E_S^{n}|\mathcal{C},\mathcal{W})} \nonumber \\
%  &=&  H(L) - H(L|Z^n,S) \nonumber \\
  &=& 
  H(L|\mathcal{C},\mathcal{W}) + H(E_S^{n}|Z^n,\mathcal{C},\mathcal{W}) 
  - H(L,E_S^{n}|Z^n,\mathcal{C},\mathcal{W}) \nonumber \\
  &=& 
  H(L|\mathcal{C},\mathcal{W}) + H(E_S^{n}|Z^n,\mathcal{C},\mathcal{W}) 
  + H(\tilde Y^n|Z^n,L,E_S^{n},\mathcal{C},\mathcal{W})
  - H(L,E_S^{n},\tilde Y^n|Z^n,\mathcal{C},\mathcal{W}) \nonumber \\
  &\leq& 
  H(L|\mathcal{C},\mathcal{W}) + H(E_S^{n}|\mathcal{C},\mathcal{W}) 
  + H(\tilde Y^n|Z^n,L,E_S^{n}) - H(\tilde Y^n|Z^n,\mathcal{C},\mathcal{W}) \nonumber \\
  &=& 
  H(L|\mathcal{C},\mathcal{W}) + H(E_S^{n}|\mathcal{C},\mathcal{W}) 
  + H(\tilde Y^n|Z^n,L,E_S^{n}) - H(\tilde Y^n) + I(\tilde Y^n;Z^n),
\end{eqnarray}
where 
%the inequality follows from the property that conditioning
%reduces entropy, and 
the last equality follows from the fact that $(\tilde Y^n,Z^n)$ is
independent of $(\mathcal{C},\mathcal{W})$.  Also $I(\tilde Y^n;Z^n) =
nI(\tilde Y;Z) = nC_w(\tilde\beta)$ because of the memoryless nature
of the channel from $\tilde Y^n$ and $Z^n$ and of the fact that the
$\Pr(\tilde Y=+1)=\Pr(\tilde Y=-1)=0.5$ achieves the capacity of this
channel.  Moreover, consider a fictitious receiver at wiretapper
trying to decode for $\tilde Y^n$ from observing $Z^n$, $E_S^{n}$, and
$\hat X^n_0$ (or $L$ equivalently). Suppose that the ensemble average
error probability achieved by this receiver, employing ML decoding, is
$\bar \epsilon_w$. Then we have $H(\tilde Y^n|Z^n,L,E_S^{n}) \leq
1+(l-k)\bar \epsilon_w$ again by Fano's inequality. Putting all these
and~\eqref{eqn:ILZS} back into~\eqref{eqn:IKZS}, we obtain
\begin{eqnarray} \label{eqn:IKZgSf}
\frac{1}{n}I(K;Z^n | E_S^{n},\mathcal{C},\mathcal{W})
&\leq&
\frac{1}{n}I(K;Z^n,E_S^{n},\mathcal{C},\mathcal{W}) \nonumber \\
&\leq& 
C_w(\tilde\beta) - (R_c-R_k) + R_k\bar \epsilon_s + (R_c-R_k)\bar \epsilon_w +
  \frac{2}{n}.
\end{eqnarray}
%We note that~\eqref{eqn:IKZgSf} provide an upper bound of key
%leakage rate.

The preceding secrecy analysis of the proposed secret sharing scheme
based on the secret sharing regular LDPC code ensembles allow us to
arrive at the following result:
\begin{theorem}\label{thm:relax_keyach}
  Fix $\tilde\beta>0$.  Suppose that $C_w(\tilde\beta) \leq R_c \leq
  C_s(\tilde\beta)$. For any $R_l \geq 0$, choose $R_k =
  \min\{R_c-C_w(\tilde\beta)+R_l, R_c\}$. Then $(R_k,R_l)$ is an
  achievable key-leakage rate pair through the BPSK-constrained
  Gaussian wiretap channel with symbol-by-symbol hard-decision
  destination quantization. Moreover, this rate pair can be achieved
  by the aforementioned secret sharing scheme using the secret sharing 
  $(d_v,d_c)$-regular LDPC code ensemble described before when $n$
  increases.
\end{theorem}
\begin{proof}
  First, suppose that $R_c < C_s(\tilde\beta)$ and $R_l>0$.  Since
  $R_c \geq C_w(\tilde\beta)$, $R_k >0$.  Then $R_c-R_k =
  \max\{C_w(\tilde\beta)-R_l,0\} < C_w(\tilde\beta)$.  Thus, by
  \eqref{eqn:IKZgSf}, if we can show that there is a pair $(d_v,d_c)$
  such that $R_c=1-\frac{d_v}{d_c}$, and both $\bar\epsilon_s$ and
  $\bar\epsilon_w$ in the preceding discussion vanish as $n$
  increases, then Condition 3 in Section~\ref{sec:relax_keycap} will
  be satisfied when $n$ is sufficiently large.  From the preceding
  discussion, Conditions 1, 2, and 5 will also be satisfied.
  Comparing \eqref{eqn:C4} and Condition 4, we see then that
  $(R_k,R_l)$ will be an achievable key-leakage pair. The existence of
  such pair $(d_v,d_c)$ results from the following lemma, whose proof
  is an adaptation of the arguments in \cite[Theorem 3]{Miller2001} to
  the proposed secret sharing $(d_v,d_c)$-regular LDPC code
  ensemble. The details are presented in Appendix~\ref{app:2}.
\begin{lemma}\label{lm:existldpc}
  Consider the ensemble average error probabilities $\bar \epsilon_w$
  and $\bar \epsilon_s$ achieved by the respective ML decoders at the
  source and wiretapper of the secret sharing $(d_v,d_c)$-regular LDPC
  code ensemble mentioned above. For any fixed $\tilde\beta>0$,
  suppose that $R_c < C_s(\tilde\beta)$ and $R_c - R_k <
  C_w(\tilde\beta)$. Then, there exists a choice of $(d_v,d_c)$ such
  that
  \begin{enumerate}
  \item $R_c = 1 - \frac{d_v}{d_c}$,
  \item $\bar \epsilon_w$ decreases exponentially (polynomially) with
    increasing $n$ for $R_k>0$ (for $R_k=0$), and
  \item $\bar \epsilon_s$ decreases polynomially with increasing $n$.
  \end{enumerate}
\end{lemma}
Finally, note that the before-imposed restrictions $R_c <
C_s(\tilde\beta)$ and $R_l>0$ can be removed since the key-leakage
rate region is closed.
\end{proof}
A comparison of Theorem~\ref{thm:relax_keyach} and
\eqref{eqn:relaxkeycap_bpsk_constr} shows that the restriction to the
secret sharing regular LDPC code ensemble described in this section
does not reduce the relaxed key capacity of the BPSK-constrained
Gaussian wiretap channel with destination hard-decision quantization.

As mentioned in \autoref{sec:intro}, a similar LDPC-based secret-key
agreement scheme employing observations of correlated discrete
stationary sources at the source, destination, and wiretapper was
studied in~\cite{Muramatsu2006}. After Step 1) of our proposed secret
sharing scheme, the observations $X^n$, $\tilde Y^n$, and $Z^n$ at the
three terminals can be viewed as generated from correlated sources;
thus reducing our model to the one considered
in~\cite{Muramatsu2006}\footnote{Our destination and source correspond
  to the sender and receiver in~\cite{Muramatsu2006},
  respectively. For convenience, we employ our terminology here when
  referring to the scheme in~\cite{Muramatsu2006}.}, except that the
wiretapper alphabet is continuous in our case.  As in our scheme, the
scheme in~\cite{Muramatsu2006} has the syndrome $S^{n-l}$ of $\tilde
Y^n$ sent to the source. On the other hand, the key
in~\cite{Muramatsu2006} is obtained by calculating the syndrome of
$\tilde Y^n$ with respect to another independently selected LDPC
code. The scheme in~\cite{Muramatsu2006} is shown to achieve key
capacity via a similar approach as ours. First, the consideration of
leakage information is converted to that of the error probabilities
achieved by decoders at the source and wiretapper by an upper bound
similar to~\eqref{eqn:IKZgSf} for a pair of fixed LDPC codes
(cf. Eqn.~\eqref{eqn:IKS}). Then, the existence of a fixed code pair
with vanishing error probabilities is shown via an ML decoding error
analysis of the code ensemble based on the method of
types~\cite{Bennatan2004}. Because of the continuous wiretapper
alphabet, the ML decoding error analysis in~\cite{Muramatsu2006} does
not directly apply to our case. Hence, we have opted for the combined
union and Shulman-Feder bounding technique in~\cite{Miller2001}, which
does however require the BISO nature of the channel from the
(quantized) destination to the wiretapper.  Obviously,
Lemma~\ref{lm:existldpc} also implies the existence of a fixed
$(\mathcal{C},\mathcal{W})$ from the secret sharing regular LDPC
ensemble with vanishing decoding errors in our design, and hence the
use of this fixed $(\mathcal{C},\mathcal{W})$ is also sufficient to
achieve the relaxed key capacity in our case.

Expressed in our notation, elements in the LDPC code ensemble of
\cite{Muramatsu2006} are also of the form $(\mathcal{C},\mathcal{W})$.
For our ensemble, $\mathcal{W}$ is (conditionally) uniformly
distributed over the set of all subspaces of a given
$\mathcal{C}$. For the ensemble of \cite{Muramatsu2006}, $\mathcal{W}$
is (conditionally) uniformly distributed over the set of subspaces of
$\mathcal{C}$ specified by the concatenation of the parity matrices of
$\mathcal{C}$ and another properly chosen regular LDPC code.  While
each element in the ensemble of \cite{Muramatsu2006} is also an
element of our ensemble, the two ensembles are different since the
respective (conditional) uniform distributions for $\mathcal{W}$ are
defined over two different sets of subspaces. In a sense, the ensemble
of \cite{Muramatsu2006} is more restrictive since $\mathcal{W}$ also
needs to be an LDPC code. The discussion in this section shows that
the LDPC structure needs to be imposed only on $\mathcal{C}$ but not
on $\mathcal{W}$. This bears significance in the design of practical
codes because the design based on one LDPC structure derived from our
ensemble is much simpler, as will be illustrated in the following
section.

%In summary, the main difference between our secret sharing scheme and
%the one in~\cite{Muramatsu2006} lies in the choice of key generation
%method from $\tilde Y^n$. Rather than generating the key as the
%syndrome of another independently chosen LDPC code, we have the key as
%the index of codeword in the quotient by a uniformly chosen subspace
%$\mathcal{W}$ of the LDPC code $\mathcal{C}$.  The question of course
%is how to find such a choice under practical constraints such as
%finite block length and finite decoding complexity.  We believe the
%main innovation of the proposed secret sharing scheme with respect to
%the one suggested in~\cite{Muramatsu2006} is that our secret sharing
%code construction is much more conducive to design under practical
%constraints. This will be .

\section{Secret sharing scheme employing practical LDPC codes}
\label{sec:wz}

In practice, it is not realistic to employ the secret sharing regular
LDPC code ensemble and ML decoding at the source as suggested in
\autoref{sec:ssldpc}, for even moderate values of $n$. In this
section, we investigate the secrecy performance of a secret sharing
scheme similar to the one suggested in \autoref{sec:ssldpc}, but with
fixed choices of $(\mathcal{C},\mathcal{W})$ from the secret sharing
regular LDPC code ensemble and more-practical BP decoding. In
addition, from the proof of Lemma~\ref{lm:existldpc} in
Appendix~\ref{app:2}, the values of $d_v$ and $d_c$ need to be large
in order for the ensemble average error probabilities $\bar
\epsilon_w$ and $\bar \epsilon_s$ to decrease with $n$, and hence to
achieve the relaxed key capacity. As large values of $d_v$ and $d_c$
increase the graph complexity of a LDPC code, and hence the complexity
of BP decoding, we have to limit ourselves to small values of $d_v$
and $d_c$.  To alleviate the shortcoming of regular LDPC codes with
small $d_v$ and $d_c$, we also consider the use of more-efficient
irregular LDPC codes in the proposed secret sharing scheme.

We consider the secret sharing scheme described in
\autoref{sec:ssldpc}, except that the secret sharing code
$(\mathcal{C},\mathcal{W})$ is fixed and is known to the source and
destination (and also the wiretapper) beforehand. Here, we consider
the (fixed) code $\mathcal{C}$ chosen from ensembles of regular and
irregular LDPC codes.  The details will be discussed later.  For
convenience in the key generation step (and later in the search of
good irregular LDPC codes), the subspace $\mathcal{W}$ is chosen as
follows. Referring back to Step 2) of the scheme, choose a lower
triangular version\footnote{We can, without loss of generality, assume
  $H$ to be of full rank as discussed before.  Alternatively, an
  approximate lower triangular version of $H$ as described
  in~\cite{RichardsonIT01_3} can also be used if efficient encoding is
  needed.} of $H$, for example by performing Gaussian elimination on
the connection matrix of the bipartite graph of $\mathcal{C}$ as
discussed in~\cite{RichardsonIT01_3}.  Hence, $H=[A,B]$ where $B$ is
an $(n-l)\!\times\!  (n-l)$ lower triangular matrix.  Write $\tilde
Y^n=[d^l, e^{n-l}]$ where $d^l$ and $e^{n-l}$ are row vectors
containing $l$ and $n-l$ elements, respectively.  Then the syndrome
$S^{n-l}=d^lA^T+e^{n-l}B^T$, codeword $X_0^n = [d^l,
d^lA^T(B^{-1})^T]$ and coset leader $E_S^n = [0^T,
S^{n-l}(B^{-1})^T]$. Note that $d^l$ contains the systematic bits of
the codeword $X_0^n$ while $d^lA^T(B^{-1})^T$ contains the parity
bits. The subspace $\mathcal{W}$ is chosen to be the set of codewords
obtained by setting the first $k$ bits\footnote{It is easy to see that
  the secrecy performance is the same for any choice of $k$ bits in
  $d^l$ for the BP decoders described below.} in the vector $d^l$
above to zero. The quotient space $\mathcal{K}$ is isomorphic to the
set of codewords obtained by setting the last $l-k$ bits in the vector
$d^l$ to zero.  Hence we can use the first $k$ bits in $d^l$ as the
key. Since $(\mathcal{C},\mathcal{W})$ is known to the source
beforehand, there is no need to feed it back to the source via the
public channel in Step 2) of the secret sharing scheme. Step 3) of the
scheme is modified to replace ML decoding by the practical BP
decoding.

First, it is unlikely that the above fixed choice of $\mathcal{W}$
results in an LDPC code. Hence, the fixed coding scheme suggested here
is different from that of \cite{Muramatsu2006}. Second, the secrecy
analysis of \autoref{sec:ssldpc} can be easily modified to reflect the
use of the fixed secret sharing code $(\mathcal{C},\mathcal{W})$
mentioned above. In particular, the upper bound on the leakage rate in
\eqref{eqn:IKZgSf} becomes
\begin{equation} \label{eqn:IKS} \frac{1}{n}I(K;Z^n | E_S^{n}) \leq
  C_w(\tilde \beta) - (R_c-R_k) + R_k\epsilon_s + (R_c-R_k)\epsilon_w
  + \frac{2}{n},
\end{equation}
where $\epsilon_s$ and $\epsilon_w$ are now the error probabilities
achieved by the BP decoders at the source and wiretapper,
respectively. Since the bound above is derived from Fano's inequality,
it applies for any decoder (ML, BP, \emph{etc.}), and the value of the
bound depends on the choices of decoders only through $\epsilon_s$ and
$\epsilon_w$. Below, we perform computer simulation to estimate
$\epsilon_s$ and $\epsilon_w$ and then employ \eqref{eqn:IKS} to bound
the leakage rates achieved by $(\mathcal{C},\mathcal{W})$ constructed
from different choices of finite block length LDPC codes as described
above. More specifically, suppose that the key rate of a secret
sharing LDPC code $(\mathcal{C},\mathcal{W})$ is $R_k$ and
$\epsilon_s$ obtained from simulation is small. By setting $R_l$ to be
the value of the bound~\eqref{eqn:IKS} obtained as described, then
$(R_k,R_l)$ will be considered a key-leakage rate pair achievable by
$(\mathcal{C},\mathcal{W})$.

\subsection{Secret sharing regular LDPC codes}\label{subsec:regldpc}

We start by evaluating the secrecy performance of using regular LDPC
codes with small $d_v$ and $d_c$ in the secret sharing scheme
described above.  First, we pick $\mathcal{C}$ from the rate-$0.25$
$(3,4)$-regular LDPC code ensemble by realizing the random bipartite
graph experiment described in~\cite{RichardsonIT01_1} and then remove
all length-$4$ loops in the realization. The block length $n$ of the
LDPC code is set to $10^5$. As mentioned above, we need to estimate
the values of $\epsilon_s$ and $\epsilon_w$ from computer simulation.
To get $\epsilon_s$, BP decoding is implemented at the
source. Similarly, a BP decoder is implemented for the fictitious
receiver at the wiretapper to obtain $\epsilon_w$.  In order to
provide information about $L$ to the latter decoder, the intrinsic
log-likelihood ratios (LLRs) of the first $k$ elements in $d^l$, which
are associated with $L$, are explicitly set to $\pm \infty$ according
to the true bit values.  While this method may not be the optimal way
to feed information of $L$ to the BP decoder, we choose to employ it
because of its simplicity and the fact that this method also allows
simple density evolution analysis, which will be used to search for
good irregular LDPC codes in \autoref{subsec:irregldpc} below.

\autoref{fig:ck_rl_regular} shows the trajectory of $(R_k,R_l)$
achievable by the rate-$0.25$ secret sharing $(3,4)$-regular LDPC code
$(\mathcal{C},\mathcal{W})$ when the maximum allowable SNR
$P/\sigma^2$ is limited to $-0.15$~dB and $\alpha^2 =0$~dB. Different
values of $R_k$ on the trajectory shown are obtained by varying the
value of $k$ (i.e., the dimension of $\mathcal{W}$ also changes).
When obtaining each shown pair $(R_k,R_l)$, we choose $\tilde
\beta^2$, up to $P/\sigma^2$, such that $ \epsilon_s \leq 0.01$,
$\epsilon_w \leq 0.01$ and the bound in~\eqref{eqn:IKS} is
minimized. For any so-obtained pair $(R_k,R_l)$ located to the right
of the $45^{\circ}$ line in Fig.~\ref{fig:ck_rl_regular}, the bound
\eqref{eqn:IKS} becomes too loose, and the pair is not plotted. From
\autoref{fig:ck_rl_regular}, we observe that the pair
$(R_k,R_l)=(0.2,0.139)$ gives the smallest (bound on) leakage rate
that is achievable by the rate-$0.25$ secret sharing $(3,4)$-regular
LDPC code in the proposed scheme.

Next, we try to compare the secrecy performance of our secret sharing
scheme to that of~\cite{Muramatsu2006}.  As discussed near the end of
\autoref{sec:ssldpc}, the scheme of~\cite{Muramatsu2006} requires a
pair of independently chosen regular LDPC codes.  Since no practical
code designs or examples are provided in~\cite{Muramatsu2006}, we
choose an LDPC code pair for the scheme of~\cite{Muramatsu2006} that
is similar to the choice of our secret sharing code above for
comparison. For the scheme of~\cite{Muramatsu2006}, the first LDPC
code is set to be $\mathcal{C}$ above (i.e., the rate-$0.25$
$(3,4)$-regular LDPC code). The other code $\mathcal{C}'$ (from which
the secret key is generated) is chosen independently from another
regular LDPC code ensemble such that a desired key rate $R_k$ is
resulted (cf. \cite{Thangaraj2007}). Note that only a few values of
$R_k$ are possible if $d_v$ and $d_c$ are restricted to have small
values. Again, as discussed near the end of \autoref{sec:ssldpc}, the
pair $(\mathcal{C},\mathcal{C}')$ can be expressed in our
$(\mathcal{C},\mathcal{W})$ notation. As such, the LDPC subcode
$\mathcal{W}$ is obtained from concatenating parity-check matrices of
$\mathcal{C}$ and $\mathcal{C}'$. Note that $\mathcal{W}$ is in
general an irregular LDPC code. To clearly distinguish between our
scheme and the one of~\cite{Muramatsu2006} in the discussion below, we
will employ the notation $(\mathcal{C},\mathcal{C}')$ when referring
to the latter. The bound~\eqref{eqn:IKS} is employed to determine the
rate pairs $(R_k,R_l)$ that can be achieved by
$(\mathcal{C},\mathcal{C}')$, same as described before.

Under the parameter setting above ($P/\sigma^2=-0.15$~dB, $\alpha^2
=0$~dB, and $n=10^{5}$), we are not able to find a choice of
$\mathcal{C}'$ (with small $d_v$ and $d_c$) that satisfies the
requirement $\epsilon_w \leq 0.01$. In order to illustrate the
comparison between the two schemes, we increase the value of
$P/\sigma^2$ to $2.0$~dB. For this case, we pick $\mathcal{C}$ to be a
rate-$0.4$ $(3,5)$-regular LDPC code. The $(R_k,R_l)$-trajectory
achieved by our secret sharing scheme with $(\mathcal{C},\mathcal{W})$
is overlaid in~\autoref{fig:ck_rl_regular}. We see that the lowest
leakage rate achieved by this choice of $(\mathcal{C},\mathcal{W})$ is
at the pair $(R_k,R_l)=(0.22,0.173)$. For the scheme
of~\cite{Muramatsu2006}, picking $\mathcal{C}'$ to be an
$(1,3)$-regular LDPC code, the pair $(\mathcal{C},\mathcal{C}')$
achieves the key-leakage rate pair $(R_k,R_l)=(0.333,0.286)$ as shown
by the square symbol in~\autoref{fig:ck_rl_regular}. This value of
$R_l$ is the lowest that we can obtain from picking many different
$\mathcal{C}'$ with small $d_v$ and $d_c$.

Summarizing the above results, our secret sharing scheme outperforms
the scheme of~\cite{Muramatsu2006} when the respective code employed
in each scheme is restricted among the choices of regular LDPC codes
with small node degrees and finite block lengths.
%For example, we choose
%$\mathcal{C}_2$ to be a $(1,5)$-regular LDPC code which achieves a key
%rate $R_k$ of $0.2$. Using this choice of $\mathcal{C}_1$ and
%$\mathcal{C}_2$, we need to increase $P/\sigma^2$ to $1.2$~dB and we
%observe that the pair $(R_k,R_l)=(0.2,0.234)$ achieved by the two
%regular LDPC codes locates to the right of the diagonal line in
%Fig.~\ref{fig:ck_rl_regular}. Similar observations can be made when we
%choose $\mathcal{C}_2$ to be a $(2,9)$- and $(3,13)$-regular LDPC
%codes, which results in a key rate of $0.222$ and $0.231$
%respectively. This shows that the propose scheme achieves better
%secrecy performance than that of~\cite{Muramatsu2006} by using regular
%LDPC codes.
%For the scheme in~\cite{Muramatsu2006} to
%achieve a key rate around $0.2$, we choose $\mathcal{C}_2$ such that
%$\mathcal{C}_3$ results in a rate-$0.2$ $(4,5)$ regular LDPC code. For
%this code (along with the original rate-$0.4$ regular LDPC code) to
%work (i.e., $ \epsilon_s \leq 0.01$ and $\epsilon_w \leq 0.01$),
%$\alpha^2$ need to be increased to $1.0$~dB. As seen from
%Fig.~\ref{fig:ck_rl_regular}, the pair $(R_k,R_l)=(0.2,0.191)$ is
%achievable by the two aforementioned regular LDPC codes. Even under
%such a scenario, Fig.~\ref{fig:ck_rl_regular} shows that our scheme
%still gives better secrecy performance by the pair
%$(R_k,R_l)=(0.162,0.151)$.
%Fig.~\ref{fig:ck_rl_regular} shows that the proposed scheme gives
%better secrecy performance than that of~\cite{Muramatsu2006} with the 
%use of regular LDPC codes.
However, we can observe that there is a significant gap between the
$(R_k, R_l)$ pairs achieved by the proposed scheme and the maximally
achievable $(C_{bq},R_l)$ key-leakage pair boundary. This illustrates
that regular LDPC codes with small $d_v$ and $d_c$ and finite block
length do not provide good secret sharing performance. 
% We note that the performance of the scheme in~\cite{Muramatsu2006}
% can be improved by considering the use of irregular LDPC codes for
% $\mathcal{C}_1$ and $\mathcal{C}_2$, however, no practical design
% procedure has been given in~\cite{Muramatsu2006} as mentioned
% previously.  We will present the design of irregular LDPC codes for
% use in the proposed scheme in the next section to achieve better
% secrecy performance.

\subsection{Secret sharing irregular LDPC
  codes} \label{subsec:irregldpc}
To improve secret sharing performance, we search for ``good''
irregular LDPC codes to be used as $\mathcal{C}$ in the proposed
scheme.  The structure of secret sharing code
$(\mathcal{C},\mathcal{W})$ described in the beginning of this section
facilitates the code search process because only the LDPC structure of
$\mathcal{C}$ needs to be optimized. Such optimization can be
performed by employing the density-evolution based linear programming
technique suggested in~\cite{chung2001dld}. The search objective is to
find an irregular LDPC secret-sharing code $(\mathcal{C},\mathcal{W})$
with maximum $R_c$, given a fixed $R_k$ such that both the decoding
error probabilities $\epsilon_s$ and $\epsilon_w$ in \eqref{eqn:IKS}
are vanishing as the BP decoders iterate. By \eqref{eqn:IKS}, this
results in minimization of the bound on $R_l$ for the fixed $R_k$.

Using standard notation, let the variable and check node degree
distribution polynomials of an irregular LDPC code ensemble be,
respectively, $\lambda(x) = \sum_{i=2}^{d_v} \lambda_i x^{i-1}$ and
$\rho(x) = \sum_{i=2}^{d_c} \rho_i x^{i-1}$, where $\lambda_i
(\rho_i)$ represents the fraction of edges emanating from the variable
(check) nodes of degree $i$.  We are to design an irregular LDPC code
$\mathcal{C}$ and its subcode $\mathcal{W}$ that work well for the
channel from the (quantized) destination to source and the channel
from the (quantized) destination to wiretapper, corresponding to the
error probabilities $\epsilon_s$ and $\epsilon_w$, respectively.  Fix
$\rho(x)$, and let $e_s(\ell)$ and $e_w(\ell)$ denote the bit error
probabilities obtained by the BP decoders at the source and
wiretapper, respectively, at the $\ell$th density evolution
iteration~\cite{RichardsonIT01_1,chung2001dld} when an initial
$\tilde\lambda(x) = \sum_{i=2}^{d_v} \tilde\lambda_i x^{i-1}$ is used.
Now, let $A_{\ell,j}$ denote the bit error probability obtained at the
source by running the density evolution for $\ell$ iterations, in
which $\tilde\lambda(x)$ is used as the variable node degree
distribution for the first $\ell - 1$ iterations and the variable node
degree distribution with a singleton of unit mass at degree $j$ is
used for the final iteration. Let $B_{\ell,j}$ denote the similar
quantity for bit error probability obtained at the wiretapper. Then,
we have $e_s(\ell)=\sum_{j=2}^{d_v} A_{\ell,j} \tilde\lambda_j$ and
$e_w(\ell)=\sum_{j=2}^{d_v} B_{\ell,j} \tilde\lambda_j$.  Note that
the values of $A_{\ell,j}$ and $B_{\ell,j}$ are obtained via density
evolution. To account for the availability of perfect information of
the $k$ bits corresponding to the key at the wiretapper's BP decoder,
the intrinsic LLR distribution entered into the density evolution
analysis for the wiretapper's decoder is set to be a mixture of the
distribution of the channel outputs at the wiretapper (with the
quantized destination symbols as the channel input) and an impulse at
$+\infty$. The weights of the two components in the mixture are
determined by the value of $R_k$.

Let $\epsilon>0$ be a small prescribed error tolerance. Suppose that
$\tilde \lambda(x)$ satisfies the property that $e_s(M_s) \leq
\epsilon$ and $e_w(M_w) \leq \epsilon$, for some integers $M_s$ and
$M_w$. Then, we can frame the $R_c$-maximizing code design problem as
the following linear program:
\begin{eqnarray*}
  && \max_{\lambda(x)} \sum_{j=2}^{d_v} \frac{\lambda_j}{j} \\
  \mbox{subject to} && \\
  &&\sum_{j=2}^{d_v} \lambda_j = 1, \mbox{~~~~~~~~~~~~}
  \lambda_i \geq 0  \mbox{~~for~}2 \leq i \leq d_v \mbox{,} \\
  &&\left |\sum_{j=2}^{d_v} A_{\ell,j} \lambda_j - e_s(\ell)\right|
  \leq \max[0,\delta (e_s(\ell-1)-e_s(\ell))], \mbox{~~~~~for~}1 \leq \ell \leq M_s \\
  &&\left |\sum_{j=2}^{d_v} B_{\ell,j} \lambda_j - e_w(\ell)\right| \leq \max[0,\delta (e_w(\ell-1)-e_w(\ell))],  \mbox{~~~~~for~}1 \leq \ell \leq M_w \\
  &&\sum_{j=2}^{d_v} A_{\ell,j} \lambda_j \leq e_s(\ell-1),  \mbox{~~~~~for~}1 \leq \ell \leq M_s \\
  &&\sum_{j=2}^{d_v} B_{\ell,j} \lambda_j \leq e_w(\ell-1),  \mbox{~~~~~for~}1 \leq \ell \leq M_w 
\end{eqnarray*}
where $d_v$ here is the maximum allowable degree of $\lambda(x)$ and
$\delta$ is a small positive number.  The solution $\lambda(x)$ of the
above linear program is then employed as the initial
$\tilde\lambda(x)$ for the next search round. The search process
continues this way until $e_s(M_s)$ or $e_w(M_w)$ becomes larger than
$\epsilon$, or until $\lambda(x)$ converges. We can also fix
$\lambda(x)$ and obtain a similar linear programming problem for
$\rho(x)$. The iterative search can then alternate between the linear
programs for $\lambda(x)$ and $\rho(x)$, respectively.

The secret sharing irregular LDPC codes presented below are obtained
from the code search procedure described above starting with
BSC-optimized LDPC codes, which are available from Urbanke's
website~\cite{Urbanke01}.  \autoref{fig:ck_rl_irregular_0dB} shows the
$(R_k,R_l)$-trajectory achieved by a rate-$0.25$ secret sharing
irregular LDPC code obtained by performing the above search with $R_k$
set to $0.155$ for the BPSK-constrained Gaussian wiretap channel when
$P/\sigma^2=-1.5$~dB and $\alpha^2=0$~dB.  The degree distribution
pair of this secret sharing irregular LDPC code is shown in
\autoref{tb:deg_dist_pair}.  We obtain an instance of the irregular
code by randomly generating a bipartite graph which satisfies the two
given degree-distribution constraints. Similar to the case of regular
codes, the block length $n=10^5$, and all length-$4$ loops are
removed.  Each shown $(R_k,R_l)$ pair is obtained in the same manner
as described in \autoref{subsec:regldpc} by using
\eqref{eqn:IKS}. From~\autoref{fig:ck_rl_irregular_0dB}, we observe
that the pair $(R_k,R_l)=(0.155,0.025)$ gives the lowest leakage rate
achievable by this secret sharing irregular LDPC code. For comparison,
we also plot in~\autoref{fig:ck_rl_irregular_0dB} the
$(R_k,R_l)$-trajectory achieved by the proposed secret sharing scheme
using a rate-$0.25$ BSC-optimized irregular LDPC code in place of the
secret sharing irregular LDPC code obtained from the code search
described above.  Note that since the channel from the (quantized)
destination to the source is a BSC, the use of the BSC-optimized LDPC
code is essentially the same as the reconciliation method proposed
in~\cite{Elkouss2009}.  For the BSC-optimized code, the pair
$(R_k,R_l)=(0.2,0.071)$ gives the lowest achievable leakage rate.

Similarly,~\autoref{fig:ck_rl_irregular_5dB} shows the secrecy
performance of the proposed scheme when $P/\sigma^2 = -4.9$~dB and
$\alpha^2 = 5$~dB. A rate-$0.12$ secret sharing irregular LDPC code is
obtained by fixing $R_k$ to $0.06$ in the code search.  The degree
distribution pair of this secret sharing irregular LDPC code is also
shown in Table~\ref{tb:deg_dist_pair}.  We observe that the lowest
leakage rate achieved by this code is given by the pair
$(R_k,R_l)=(0.062,0.019)$. Again, for comparison, the
$(R_k,R_l)$-trajectory achieved by replacing the secret sharing
irregular LDPC code obtained from the code search with a rate-$0.12$
BSC-optimized irregular LDPC code is also shown
in~\autoref{fig:ck_rl_irregular_5dB}. For the BSC-optimized irregular
LDPC code, the pair $(R_k,R_l)=(0.095,0.052)$ gives the lowest
achievable leakage rate. In conclusion, the secret sharing irregular
LDPC codes obtained from the proposed code search procedure
significantly outperform, in terms of secrecy performance, secret
sharing regular LDPC codes with small node degrees as well as
irregular LDPC codes that are optimized just for information
reconciliation.

\section{Conclusions}\label{sec:con}
In this paper, we developed schemes based on LDPC codes to allow a
source and a destination to share secret information over a
BPSK-constrained Gaussian wiretap channel. In the proposed
secret sharing schemes, the source first sends a random BPSK symbol
sequence to the destination through the Gaussian wiretap
channel. Then, the destination generates a syndrome of its quantized
received sequence using an LDPC code and sends this syndrome back to
the source via the public channel. Finally, the source performs
decoding to recover the quantized destination sequence based on its
transmitted sequence, as well as the syndrome that it receives from
the destination.  The secret key is obtained as the index of a coset
in a quotient space of the LDPC code.

To evaluate the performance of the proposed secret sharing scheme, we
employed an upper bound on the leakage information rate that depends
on the decoding error probabilities of the decoder at the source and
of a fictitious decoder at the wiretapper, which observes the
wiretapper received sequence, the syndrome in the public channel as
well as the secret key. The design was then converted to making these
error probabilities small. For a suitably chosen ensemble of regular
LDPC codes, we showed that these error probabilities can indeed be
made vanishing, as the block length increases, by ML decoding. As a
result, this established that the key capacity of the BPSK-constrained
Gaussian wiretap channel can be achieved by employing the secret
sharing regular LDPC code ensemble in the proposed scheme.

Considering the practical constraints of finite block length and using
BP decoding instead of ML decoding, we employed a density-evolution
based linear program to search for good irregular LDPC codes that can
be used in the secret sharing scheme. Simulation results showed that
the secret sharing irregular LDPC codes obtained from our search can
get relatively close to the relaxed key capacity of the
BPSK-constrained Gaussian wiretap channel, significantly outperforming
regular LDPC codes as well as irregular LDPC codes that are optimized
just for information reconciliation.

Finally, we point out that the arguments in the proof of
\autoref{thm:relax_keyach} can be modified to show the existence of an
LDPC code (from the same regular LDPC code ensemble considered in
\autoref{sec:ssldpc}) that achieves the secrecy
capacity~\cite{Wyner1975,Leung1978} of the Gaussian wiretap channel
with the BPSK source-symbol constraint. The code search approach
described in \autoref{subsec:irregldpc} can also be employed to find
irregular LDPC codes that give secrecy performance close to the
boundary of the secrecy-equivocation rate region of that channel.

\bigskip
\appendices
\section{Sketch of Proof of
  Theorem~\ref{thm:keycapgen}}\label{app:pf_keycapgen}

The proof of~\cite[Theorem 2.1]{Wong2009}, which corresponds to the
case when $R_l=0$, can be easily extended to accommodate Conditions 2
and 3 in the definition of achievable key-leakage rate pair.

First, consider the converse proof. Any permissible secret sharing
strategy that achieves the key-leakage rate pair $(R,R_l)$ must
satisfy (cf. \cite[Eqn. (7)]{Wong2009})
\begin{equation}\label{eqn:IKL}
R < \frac{1}{1-\varepsilon} \left[\frac{1}{n} I(K;L) + \frac{1}{n} +
\varepsilon^2\right] + \varepsilon.
\end{equation}
From Conditions 2, 3, and the chain rule, we have
\begin{eqnarray*}
\frac{1}{n} I(K;L) &\leq& %\frac{1}{n} I(K;L,Z^n,\Phi^t,\Psi^t) \\
\frac{1}{n} I(K;L|Z^n,\Phi^t,\Psi^t) + \frac{1}{n}
I(K;Z^n|\Phi^t,\Psi^t)
+ \frac{1}{n} I(K;\Phi^t,\Psi^t) \\
&\leq&
\frac{1}{n} I(K;L|Z^n,\Phi^t,\Psi^t) + R_l + 2\varepsilon 
\leq \frac{1}{n} \sum_{j=1}^{n} I(X_j;Y_j|Z_j) + R_l +
2\varepsilon,
\end{eqnarray*}
where the last inequality is due to the bound
$I(K;L|Z^n,\Phi^t,\Psi^t) \leq \sum_{j=1}^{n} I(X_j;Y_j|Z_j)$ which is
shown in~\cite[pp. 1129--1130]{Ahlswede1993}.  Similarly, using the
chain rule and Condition 2, we also have
\begin{eqnarray*}
\frac{1}{n} I(K;L) &\leq& %\frac{1}{n} I(K;L,\Phi^t,\Psi^t) \\
\frac{1}{n} I(K;L|\Phi^t,\Psi^t) + \frac{1}{n} I(K;\Phi^t,\Psi^t) \\
&\leq&
\frac{1}{n} I(K;L|\Phi^t,\Psi^t) + \varepsilon
\leq \frac{1}{n} \sum_{j=1}^{n} I(X_j;Y_j) + \varepsilon,
\end{eqnarray*}
where the last inequality is due to the bound $I(K;L|\Phi^t,\Psi^t)
\leq \sum_{j=1}^{n} I(X_j;Y_j)$, which again can be shown by a simple
modification to~\cite[pp. 1129--1130]{Ahlswede1993}.

As in \cite{Wong2009}, let $Q$ be a uniform random variable that takes
value from $\{1,2,\ldots,n\}$ and is independent of all other random
quantities. Define $(\acute X, \acute Y, \acute Z) = (X_j, Y_j, Z_j)$
if $Q=j$. Then $p_{\acute Y, \acute Z| \acute X}(\acute y, \acute
z|\acute x)=p_{Y,Z|X}(\acute y, \acute z|\acute x)$.  Combining the
two upper bounds on $\frac{1}{n} I(K;L)$ above, we have
\begin{eqnarray}\label{eqn:ubIKL}
\frac{1}{n} I(K;L) &\leq& \min\left\{ I(\acute X; \acute Y|\acute Z,
Q) + R_l,
I(\acute X;\acute Y|Q) \right\} + 2\varepsilon  \nonumber \\
&\leq& \min\left\{ I(\acute X; \acute Y|\acute Z) + R_l, I(\acute X;
\acute Y) \right\} + 2\varepsilon.
\end{eqnarray}
%where the last inequality is due to the fact that $Q \rightarrow
%\acute X \rightarrow (\acute Y,\acute Z)$ forms a Markov chain.
The power constraint (\ref{eqn:powerconstr}) implies that
%\begin{equation}\label{eqn:avpwconstr}
$E[|\acute X|^2] %\frac{1}{n} \sum_{j=1}^{n} E[|X_j|^2] 
\leq P$.
%\end{equation}
Combining~\eqref{eqn:IKL} and~\eqref{eqn:ubIKL}, we obtain
\begin{equation} \label{eqn:ubR}
  R < \frac{1}{1-\varepsilon} \left[
    \min\left\{ I(\acute X; \acute Y|\acute Z) + R_l, I(\acute X; \acute Y) \right\} + 2\varepsilon + \frac{1}{n}
  \right].
\end{equation}
Since $\varepsilon$ can be arbitrarily small, \eqref{eqn:ubR} implies
the converse result, i.e.,
\begin{eqnarray*}
R &\leq&
\min\left\{ I(\acute X; \acute Y|\acute Z) + R_l, I(\acute X; \acute Y) \right\} \\
&\leq&
\max_{X:E[|X|^2]\leq P]} \min\left\{ I(X; Y|Z) + R_l, I(X;Y) \right\} \\
&=& \max_{X:E[|X|^2]\leq P]} \min\left\{ I(X;Y)-I(Y;Z) + R_l, I(X;Y)
\right\},
\end{eqnarray*}
where the last line is due to the fact that $p(y,z|x)=p(y|x)p(z|x)$.

The achievability proof based on random Wyner-Ziv coding
in~\cite[Section 4]{Wong2009} can be used to achieve the $R_l$-relaxed
key capacity with proper modifications.  Since the code construction
statement in \cite[Section 4]{Wong2009} is rather long, we only point
out here the steps that are different for the current case due to
space limitation. The other details of the proof can be found in
\cite{Wong2009}. We also adopt the notation of \cite{Wong2009} for
easy reference.

First, fix the source distribution $p(x)$ that achieves the maximum in
the $R_l$-relaxed key capacity expression. If $R_l < I(Y;Z)$, then
modify the code construction in \cite[Section 4]{Wong2009} with the
new definitions of $R_3=I(X;\hat Y) - I(\hat Y;Z) + R_l - \varepsilon$
and $R_4=I(\hat Y;Z) - R_l - 17 \varepsilon$. Note the $p(\hat y|y)$
should be chosen to make these rates positive. The asymptotic
negligibility of $\frac{1}{n}I(K;J)$ conditioned on the code
$\mathcal{C}_n$ used in \cite[Section 4]{Wong2009} is the only
argument needed in this case that is not explicitly shown in
\cite[Section 4]{Wong2009}.  We assume below that the code
$\mathcal{C}_n$ is used. To establish that, first similar to (73) of
\cite{Wong2009} , we have
\begin{equation}\label{eqn:IKJ}
I(K;J) \leq I(L;J) + 8n\varepsilon R_3+1.
\end{equation}
by using an argument similar to that of (73) of \cite{Wong2009}. Then
for $j=1,2,\ldots,2^{nR_2}$ and $l=1,2,\ldots,2^{nR_3}$, we have
\begin{eqnarray*}
\Pr\{J=j,L=l\} &=&
\sum_{w=1}^{2^{nR_4}} \Pr\left\{M=j+(l-1)2^{nR_2}+(w-1)2^{n(R_2+R_3)}\right\}\\
&\leq& \frac{2^{-n(R_2+R_3-7\varepsilon)}}{1-\varepsilon} <
2^{-n(R_2+R_3-8\varepsilon)}
\end{eqnarray*}
for sufficiently large $n$, where the first inequality is from
\cite[Part 3 of Lemma 6]{Wong2009}. In other words, $H(J,L)>
n(R_2+R_3-8\varepsilon)$ for sufficiently large $n$. Hence, together
with the facts $H(L)<nR_3$ and $H(J)<nR_2$, we have
%\begin{eqnarray*}
\[
I(L;J) = H(L)+H(J)-H(J,L)
\leq nR_3+nR_2 - n(R_2+R_3-8\varepsilon) = 8n\varepsilon.
\]
%\end{eqnarray*}
Putting this bound back to~\eqref{eqn:IKJ}, we obtain $\frac{1}{n}
I(K;J) \leq 8\varepsilon (R_3+1)+\frac{1}{n}$.  Since $\varepsilon$
can be chosen arbitrarily, we establish the achievability of the
relaxed key capacity.  On the other hand, if $R_l \geq I(Y;Z)$, the
code construction described above can be trivially modified to achieve
the relaxed key capacity by setting $R_4=0$ and $R_3$ arbitrarily
close to $I(X;\hat Y)$.

\section{Proof of Lemma~\ref{lm:existldpc}}\label{app:2}
As mentioned in the proof of Theorem~\ref{thm:relax_keyach}, we adapt
the proof of~\cite[Theorem 3]{Miller2001} to prove this lemma. The
main argument is to establish that there is a secret sharing
$(d_v,d_c)$-regular LDPC code ensemble $(\mathcal{C},\mathcal{W})$ for
which the ensemble average error probabilities $\bar \epsilon_s$ and
$\bar \epsilon_w$ simultaneously vanish as $n$ increases under the
assumptions stated in the lemma.

To that end, we first examine the average weight spectra of the code
$\mathcal{C}$ and subspace $\mathcal{W}$ in the LDPC code ensemble:
\begin{lemma} \label{thm:weightspectra}
  Consider the ensemble of $(n,l,k)$ secret sharing code
  $(\mathcal{C},\mathcal{W})$ described in
  \autoref{sec:ssldpc}. For $0<m\leq n$, let $\bar S_m$ and $\bar
  T_m$ be the average numbers of codewords of Hamming weight $m$ in
  $\mathcal{C}$ and $\mathcal{W}$, respectively. Then, we have
\begin{eqnarray}
  \label{eqn:Sm}
  \bar S_m &=& {n \choose m} \Pr( x^n \in \mathcal{C} | w(x^n)=m)  \\
  \bar T_m &=& \frac{2^{l-k}-1}{2^l - 1} \cdot \bar S_m
  \leq 2^{-k} \bar S_m \label{eqn:Tm}
\end{eqnarray}
where $w(x^n)$ is the Hamming weight of $x^n$.
\end{lemma}
\begin{proof}
  Eqn.~\eqref{eqn:Sm}, given in~\cite{Miller2001}, is obvious. It is
  also clear from the description of the code ensemble in
  \autoref{sec:ssldpc} that
\begin{eqnarray}
\bar T_m &=&
{n \choose m} \Pr(x^n \in \mathcal{W}| x \in \mathcal{C},
w(x^n)=m) \cdot \Pr( x^n \in \mathcal{C} | w(x^n)=m) \nonumber \\
&=&
\bar S_m \cdot \Pr(x^n \in \mathcal{W}| x^n \in \mathcal{C}, w(x^n)=m).
\label{eqn:PrxWC}
\end{eqnarray}
Consider any $x_0^n \neq 0^n \in \mathcal{C}$,
\[
\Pr(x^n_0 \in \mathcal{W}| x^n_0 \in \mathcal{C}) =
\frac{\mbox{number of $(l-k)$-dimensional subspaces in $\mathcal{C}$
    that contain $x_0^n$}}{\mbox{number of $(l-k)$-dimensional
    subspaces in $\mathcal{C}$}}.
\]
The number of $(l-k)$-dimensional subspaces in $\mathcal{C}$ is
$\displaystyle \prod_{u=1}^{l-k} \frac{2^{l-u+1}-1}{2^{l-k-u+1}-1}$
(see \cite[Theorem 7.1]{Kac02}). Further, let $\mathcal{X}_0
=\{0^n,x_0^n\}$, and let $\mathcal{C}'=\mathcal{C}/\mathcal{X}_0$ be
the quotient of $\mathcal{C}$ by $\mathcal{X}_0$. Then $\mathcal{C}'$
is a $(l-1)$-dimensional linear space. If $\mathcal{W}$ is an
$(l-k)$-dimensional subspace in $\mathcal{C}$ that contains $x_0^n$,
then $\mathcal{W}'=\mathcal{W}/\mathcal{X}_0$ is an
$(l-k-1)$-dimensional subspace in $\mathcal{C}'$. On the other hand,
suppose that $\mathcal{W}'$ is an $(l-k-1)$-dimensional subspace in
$\mathcal{C}'$. Then $\mathcal{W}=\cup_{w^n+\mathcal{X}_0 \in
  \mathcal{W}'} \ w^n+\mathcal{X}_0$ is an $(l-k)$-dimensional subspace
in $\mathcal{C}$ that contains $x_0^n$. It is also easy to see that the
correspondence between $\mathcal{W}'$ and $\mathcal{W}$ above is
one-to-one.  As a result, the number of $(l-k)$-dimensional subspaces
in $\mathcal{C}$ that contain $x_0^n$ must be the same as the number of
$(l-k-1)$-dimensional subspaces in $\mathcal{C}'$, i.e.,
$\displaystyle \prod_{u=1}^{l-k-1} \frac{2^{l-u}-1}{2^{l-k-u}-1}$. So
we have
\[
\Pr(x^n_0 \in \mathcal{W}| x^n_0 \in \mathcal{C}) = \frac{2^{l-k}-1}{2^l - 1}
\]
for all $x_0^n\neq 0 \in \mathcal{C}$. This implies
\[
\Pr(x^n \in \mathcal{W}| x^n \in \mathcal{C}, w(x^n)=m) =
\frac{2^{l-k}-1}{2^l - 1} \leq 2^{-k}
\]
for $0<m\leq n$.  Putting this back into \eqref{eqn:PrxWC}, we obtain
\eqref{eqn:Tm}.
\end{proof}
For $\mathcal{C}$ chosen uniformly from the $(d_v,d_c)$-regular LDPC
code ensemble as described in \autoref{sec:ssldpc}, an upper bound on
$\Pr( x^n \in \mathcal{C} | w(x^n)=m)$ is also available in
\cite[Lemma 2]{Miller2001}:
\begin{itemize}
\item If $md_v$ is odd, $\Pr( x^n \in \mathcal{C} | w(x^n)=m) = 0$.
\item If $md_v$ is even and $md_v \leq 2(n-l)$, $\displaystyle \Pr(
  x^n \in \mathcal{C} | w(x^n)=m) \leq {n-l \choose \frac{md_v}{2}}
  \left[\frac{md_v}{2(n-l)}\right]^{md_v}$.
\item If $md_v$ is even, $\displaystyle \Pr( x^n \in \mathcal{C} |
  w(x^n)=m) \leq \left[(n-l)d_c+1\right] \cdot
  \left[\frac{1+\left(1-\frac{2m}{n}\right)^{d_c}}{2}\right]^{n-l}$.
\end{itemize}
In addition, $\Pr( x^n \in \mathcal{C} | w(x^n)=m) = \Pr( x^n \in
\mathcal{C} | w(x^n)=n-m)$ (and hence $\bar S_{n-m} = \bar S_m$) if
$d_c$ is even.

Next, we employ Lemma~\ref{thm:weightspectra} and the combined union
and Shulman-Feder bound in \cite[Theorem 1]{Miller2001} to bound $\bar
\epsilon_s$ and $\bar \epsilon_w$.  To bound $\bar \epsilon_w$,
consider the channel with $\tilde Y^n$ as input and $Z^n$ as
output. First, note that $\tilde Y^n$ contains i.i.d. equally likely
binary elements.  Hence, this channel is a memoryless BISO channel,
and is specified by the conditional pdf $p_{Z|\tilde Y}(z|\tilde y) =
p_{Z|X}(z|1) p_{X|\tilde Y}(1|\tilde y)+p_{Z|X}(z|-1) p_{X|\tilde
  Y}(-1|\tilde y)$. Since $E_S^n+\hat X^n_0+\mathcal{W}$ is a coset
and the channel is memoryless BISO, it suffices to assume $\tilde
Y^n=\tilde X^n_0 \in \mathcal{W}$. In addition, note that all possible
$\tilde X^n_0$ sequences are equally likely. Now, let $\tilde
K=\frac{6}{d_v} \ln \frac{d_v}{1-R_c}$ and
$\bar\beta=\frac{2(1-R_c)}{d_v} e^{-12-\tilde K}$. For any $\bar\beta <
\gamma < \frac{1}{2}$, applying the bound in \cite[Theorem
1]{Miller2001} to the subcode $\mathcal{W}$, the ensemble average
decoding error probability of the ML decoder at the wiretapper can be
upper-bounded as
\begin{equation}
\bar \epsilon_w \leq 
\begin{cases}
  \tau_1 + \tau_2 + 2^{-nE_r^w\left(R_c-R_k+\frac{1}{n}\log_2
      \alpha_w\right)} & \mbox{~for odd $d_c$} \\
  \tau_1 + \tau_2 + \tau_3 + \tau_4 + \tau_5 +
  2^{-nE_r^w\left(R_c-R_k+\frac{1}{n}\log_2
      \alpha_w\right)} & \mbox{~for even $d_c$},
\end{cases}
\label{eqn:mlboundew}
\end{equation}
where $\tau_1 = \sum_{m=1}^{\bar\beta n} \bar T_m D_w^m$, $\tau_2 =
\sum_{m=\bar\beta n+1}^{\gamma n} \bar T_m D_w^m$, $\tau_3 =
\sum_{m=n-\gamma n}^{n-\bar\beta n-1} \bar T_m D_w^m$, $\tau_4 =
\sum_{m=n-\bar\beta n}^{n-1} \bar T_m D_w^m$, $\tau_5 = \bar T_n D_w^n$,
$D_w=\int \sqrt{p_{Z|\tilde Y}(z|1)\cdot p_{Z|\tilde Y}(z|-1)}\,dz$,
\[
\alpha_w = 
\begin{cases}
  \max_{m\in \{\gamma n +1, \ldots, n\}} \frac{\bar T_m}{2^{l-k}-1}
  \cdot
  \frac{2^n}{{n \choose m}}& \mbox{~for odd $d_c$} \\
  \max_{m\in \{\gamma n +1, \ldots, n-\gamma n - 1\}} \frac{\bar
    T_m}{2^{l-k}-1} \cdot \frac{2^n}{{n \choose m}}& \mbox{~for even
    $d_c$},
\end{cases}
\]
and $E_r^w(R) = \max_{q}\max_{0\leq \rho\leq 1} \{ E_0^w(\rho,q) -
\rho R\}$ is the random coding error exponent with
\[
E_0^w(\rho,q) = - \log_2 \int \left[ q(1)p_{Z|\tilde
    Y}(z|1)^{1/(1+\rho)} + q(-1)p_{Z|\tilde Y}(z|-1)^{1/(1+\rho)}
\right]^{1+\rho} dz,
\]
and $q$ is the probability mass function (pmf) of the channel input
$\tilde Y$.  It is known that the optimal $q$ is $q(1)=q(-1)=0.5$.

Employing Lemma~\ref{thm:weightspectra} and the bound on
$\Pr(x^n\in\mathcal{C}|w(x^n)=m)$ that follows (see also \cite[Lemma
2]{Miller2001}), it is not hard to further bound the various terms in
\eqref{eqn:mlboundew}:
\[
\tau_1 \leq \begin{cases}
2^{-nR_k}\, n^{1-d_v/2}\, (1-R_c)^{-d_v/2}\,\frac{D_w}{1-D_w}\,\frac{(d_v/2)^{d_v}}{(d_v/2)!} & \mbox{for even $d_v$} \\
2^{-nR_k}\, n^{2-d_v}\, (1-R_c)^{-d_v}\,\frac{D_w^2}{2(1-D_w^2)}\,\frac{(d_v)^{2d_v}}{d_v!} & \mbox{for odd $d_v$},
\end{cases}
\]
\begin{eqnarray*}
\frac{\log_2 \tau_2}{n} &\leq& 
\frac{1}{n} \left\{ \log_2 n + \log_2 [(n-k)d_c+1 \right]\} - R_k \\
& & ~+
\max_{\bar\beta\leq x\leq \gamma}\left\{ x\log_2 D_w + H_2(x)+ (1-R_c)\left(\log_2[1+(1-2x)^{d_c}]-1\right)\right\},
\end{eqnarray*} 
and for even $d_c$, 
\[
\tau_4 = \sum_{m=1}^{\bar\beta n} \bar T_m D_w^m D_w^{n-2m} \leq \tau_1 D_w^{n(1-2\bar\beta)},
\]
\[
\frac{\log_2 \tau_3}{n} \leq \frac{\log_2 \tau_2}{n} + (1-2\gamma)\log_2 D_w,
\]
and
\[
\tau_6 \leq 2^{-nR_k} D_w^{n} = 2^{-n(R_k-\log_2 D_w)}.
\]
Also,
\begin{eqnarray*}
\frac{\log_2 \alpha_w}{n} 
&\leq&
\begin{cases}
\displaystyle \frac{1}{n} \left\{ 1 + \log_2 [(n-l)d_c+1 \right]\} +
(1-R_c)\max_{\gamma \leq x \leq 1} \log_2[1+(1-2x)^{d_c}] &
\mbox{~for odd $d_c$} \\
\displaystyle \frac{1}{n} \left\{ 1 + \log_2 [(n-l)d_c+1 \right]\} +
(1-R_c)\max_{\gamma \leq x \leq 1-\gamma} \log_2[1+(1-2x)^{d_c}] &
\mbox{~for even $d_c$}
\end{cases}\\
&\leq&
\frac{1}{n} \left\{ 1 + \log_2 [(n-l)d_c+1 \right]\} +
(1-R_c)\log_2[1+(1-2\gamma)^{d_c}].
\end{eqnarray*}

For bounding $\bar \epsilon_s$, note that the channel with $\tilde
Y^n$ as input and $X^n$ as output is a memoryless BSC and is specified
by the conditional pmf $p_{X|\tilde Y}(x|\tilde y) = p_{\tilde
  Y|X}(\tilde y|x)$. Again, since $E_S^n+\mathcal{C}$ is a coset and
the channel is memoryless BISO, it suffices to assume $\tilde
Y^n=X^n_0 \in \mathcal{C}$. With this identification, the resulting
bound on $\bar \epsilon_s$ follows the same line of arguments as
above, and is essentially given in \cite{Miller2001}. We summarize the
bound below for later reference:
\begin{equation}
\bar \epsilon_s \leq 
\begin{cases}
  \sigma_1 + \sigma_2 + 2^{-nE_r^s\left(R_c+\frac{1}{n}\log_2
      \alpha_s\right)} & \mbox{~for odd $d_c$} \\
  \sigma_1 + \sigma_2 + \sigma_3 + \sigma_4 + \sigma_5 +
  2^{-nE_r^s\left(R_c+\frac{1}{n}\log_2 \alpha_s\right)} & \mbox{~for
    even $d_c$},
\end{cases}
\label{eqn:mlboundes}
\end{equation}
where 
%$\sigma_1 = \sum_{m=1}^{\beta n} \bar T_m D_w^m$, $\sigma_2 =
%\sum_{m=\beta n+1}^{\gamma n} \bar T_m D_w^m$, $\sigma_3 =
%\sum_{m=n-\gamma n}^{n-\beta n-1} \bar T_m D_w^m$, $\sigma_4 =
%\sum_{m=n-\beta n}^{n-1} \bar T_m D_w^m$, $\sigma_5 = \bar T_n D_w^n$,
\[
\sigma_1 \leq \begin{cases}
n^{1-d_v/2}\, (1-R_c)^{-d_v/2}\,\frac{D_s}{1-D_s}\,\frac{(d_v/2)^{d_v}}{(d_v/2)!} & \mbox{for even $d_v$} \\
n^{2-d_v}\, (1-R_c)^{-d_v}\,\frac{D_s^2}{2(1-D_s^2)}\,\frac{(d_v)^{2d_v}}{d_v!} & \mbox{for odd $d_v$},
\end{cases}
\]
\begin{eqnarray*}
\frac{\log_2 \sigma_2}{n} &\leq& 
\frac{1}{n} \left\{ \log_2 n + \log_2 [(n-l)d_c+1 \right]\} \\
& & ~+
\max_{\bar\beta\leq x\leq \gamma}\left\{ x\log_2 D_s + H_2(x)+ (1-R_c)\left(\log_2[1+(1-2x)^{d_c}]-1\right)\right\},
\end{eqnarray*} 
and for even $d_c$, 
\[
\sigma_4 = \sum_{m=1}^{\bar\beta n} \bar T_m D_s^m D_s^{n-2m} \leq \sigma_1 D_s^{n(1-2\bar\beta)},
\]
\[
\frac{\log_2 \sigma_3}{n} \leq \frac{\log_2 \sigma_2}{n} + (1-2\gamma)\log_2 D_s,
\]
\[
\sigma_5 \leq D_s^{n} = 2^{n\log_2 D_s},
\]
and
\[
\frac{\log_2 \alpha_s}{n} \leq \frac{1}{n} \left\{ 1 + \log_2
  [(n-l)d_c+1 \right]\} + (1-R_c)\log_2[1+(1-2\gamma)^{d_c}],
\]
with $D_s=2\sqrt{p_{X|\tilde Y}(1|1)\cdot p_{X|\tilde Y}(1|-1)}$,
% + \sqrt{p_{X|\tilde Y}(-1|1)\cdot p_{X|\tilde Y}(-1|-1)}$, 
and $E_r^s(R) = \max_{q}\max_{0\leq \rho\leq 1} \{ E_0^s(\rho,q) -
\rho R\}$ is the random coding error exponent of the channel of
interest based on
\begin{eqnarray*}
E_0^s(\rho,q) 
&=& 
- \log_2 \bigg\{ \left[ q(1)p_{X|\tilde Y}(1|1)^{1/(1+\rho)} + q(-1)p_{X|\tilde Y}(1|-1)^{1/(1+\rho)} \right]^{1+\rho} \\
&& ~~~~~~~~~~~~+
\left[ q(1)p_{X|\tilde Y}(-1|1)^{1/(1+\rho)} + q(-1)p_{X|\tilde Y}(-1|-1)^{1/(1+\rho)} \right]^{1+\rho}\bigg\}.
\end{eqnarray*}

Recall that $R_c < C_s(\tilde\beta)$ and $R_c - R_k <
C_w(\tilde\beta)$. Choose $\varepsilon>0$ small enough such that $R_c
+ 2\varepsilon < C_s(\tilde\beta)$ and $R_c - R_k + 2\varepsilon <
C_w(\tilde\beta)$. For any $0 < \gamma < 0.5$, there exist large
enough $d_v$ and $d_c$ such that
\begin{enumerate}
    \item $\frac{d_v}{d_c} = 1-R_c$,
    \item $0 < \bar\beta < \gamma$,
    \item $\tilde K<\varepsilon$, and
    \item $\log_2 \left[1+(1-2\gamma)^{d_c}\right] < \varepsilon.$
\end{enumerate}
With this choice of $(d_v,d_c)$, we have
\begin{eqnarray*}
  \lefteqn{\max_{\bar\beta\leq x\leq \gamma}\left\{H_2(x)+
      (1-R_c)\left(\log_2[1+(1-2x)^{d_c}]-1\right)\right\} } \\
  &\leq&
  H_2(\gamma) + (1-R_c)\left\{\log_2[1+(1-2\bar\beta)^{d_c}]-1\right\} \\
  &\leq&
  H_2(\gamma) + (1-R_c)\left[\log_2 \left(1+e^{-2d_c\bar\beta}\right)-1\right] \\
  &\leq& H_2(\gamma) + (1-R_c)\left[\log_2
    \left(1+e^{-4e^{-12-\varepsilon}}\right)-1\right]
\end{eqnarray*}
for any $0<\gamma<0.5$, where the second inequality follows from the
inequality $1-2x < e^{-2x}$ and the last inequality follows from the
definition of $\bar\beta$. Hence, we can make
\begin{equation*}
\max_{\bar\beta\leq x\leq \gamma}\left\{H_2(x)+
(1-R_c)\left(\log_2[1+(1-2x)^{d_c}]-1\right)\right\} < 0
\end{equation*}
by choosing $\gamma$ small enough since $C_s(\tilde\beta)\leq 1$. Thus
for sufficiently large $n$, we get the following results,
\begin{enumerate}
    \item $\frac{1}{n}\log_2\tau_2 < 0$ and $\frac{1}{n}\log_2\tau_3 < 0$,
    \item $\frac{1}{n}\log_2\sigma_2 < 0$ and $\frac{1}{n}\log_2\sigma_3 < 0$,
    \item $R_c-R_k+\frac{1}{n}\log_2 \alpha_w \leq
      R_c-R_k+(1-R_c)\varepsilon + \varepsilon < C_w(\tilde\beta)$, and
    \item $R_c+\frac{1}{n}\log_2 \alpha_s \leq R_c+(1-R_c)\varepsilon
      + \varepsilon < C_s(\tilde\beta)$.
\end{enumerate}
Further, by the well known fact that the random coding exponent is
positive if its rate argument is below channel capacity, we obtain the
stated asymptotic behaviors of $\bar \epsilon_s$ and $\bar
\epsilon_w$.

\begin{table}[h]
  \caption{Degree distribution pairs of the rate-$0.25$ and
    rate-$0.12$ secret sharing irregular LDPC codes obtained from the
    code search process described in~\autoref{subsec:irregldpc}.}
  \centering
%  \begin{tabular}{|c||c|c|c|c|c|c|c|c|c|c||c|c|}
%    \hline
%    rate & $\lambda_2$ & $\lambda_3$ & $\lambda_4$ & $\lambda_7$ & $\lambda_8$ & $\lambda_{21}$ & $\lambda_{22}$ & $\lambda_{70}$ & $\lambda_{71}$ & $\lambda_{72}$ & $\rho_5$ & $\rho_6$ \\
%    \hline
%    0.25 & 0.2807 & 0.1490 & 0.0725 & 0.0599 & 0.1343 & 0.0697 & 0.0872 & 0.0006 & 0.0264 & 0.1197 & 0.4637 & 0.5363 \\
%    \hline
%    \hline
%    rate & $\lambda_2$ & $\lambda_3$ & $\lambda_5$ & $\lambda_6$ & $\lambda_{11}$ & $\lambda_{12}$ & $\lambda_{28}$ & $\lambda_{29}$ & $\lambda_{87}$ & $\lambda_{88}$ & $\rho_4$ & $\rho_5$ \\
%    \hline
%    0.12 & 0.3651 & 0.1610 & 0.1081 & 0.0540 & 0.1123 & 0.0057 & 0.0650 & 0.0403 & 0.0806 & 0.0079 & 0.9705 & 0.0295 \\
%    \hline
%  \end{tabular}
  \begin{tabular}{|c||c|c|}
    \hline
    &rate-$0.25$ & rate-$0.12$ \\
    \hline
    \hline
    $\lambda_2$ &  0.2807 & 0.3651 \\
    \hline
    $\lambda_3$ &  0.1490 & 0.1610 \\
    \hline
    $\lambda_4$ & 0.0725 & \\
    \hline
    $\lambda_5$ & & 0.1081 \\
    \hline
    $\lambda_6$ & & 0.0540 \\
    \hline
    $\lambda_7$ &  0.0599 & \\
    \hline
    $\lambda_8$ &  0.1343 & \\
    \hline
    $\lambda_{11}$ &  & 0.1123  \\
    \hline
    $\lambda_{12}$ &  & 0.0057  \\
    \hline
    $\lambda_{21}$ & 0.0697 &  \\
    \hline
    $\lambda_{22}$ & 0.0872 & \\
    \hline
    $\lambda_{28}$ &  & 0.0650 \\
    \hline
    $\lambda_{29}$ &  & 0.0403 \\
    \hline
    $\lambda_{70}$ & 0.0006 & \\
    \hline
    $\lambda_{71}$ & 0.0264 & \\
    \hline
    $\lambda_{72}$ & 0.1197  & \\
    \hline
    $\lambda_{87}$ &  & 0.0806 \\
    \hline
    $\lambda_{88}$ & & 0.0799 \\
    \hline
    \hline
    $\rho_4$ & & 0.9705 \\
    \hline
    $\rho_5$ & 0.4637 & 0.0295 \\
    \hline
    $\rho_6$ & 0.5363 & \\
    \hline
  \end{tabular}
\label{tb:deg_dist_pair}
\end{table}

\begin{figure}[H]
\centering
\begin{center}
\includegraphics[width=0.57\textwidth]{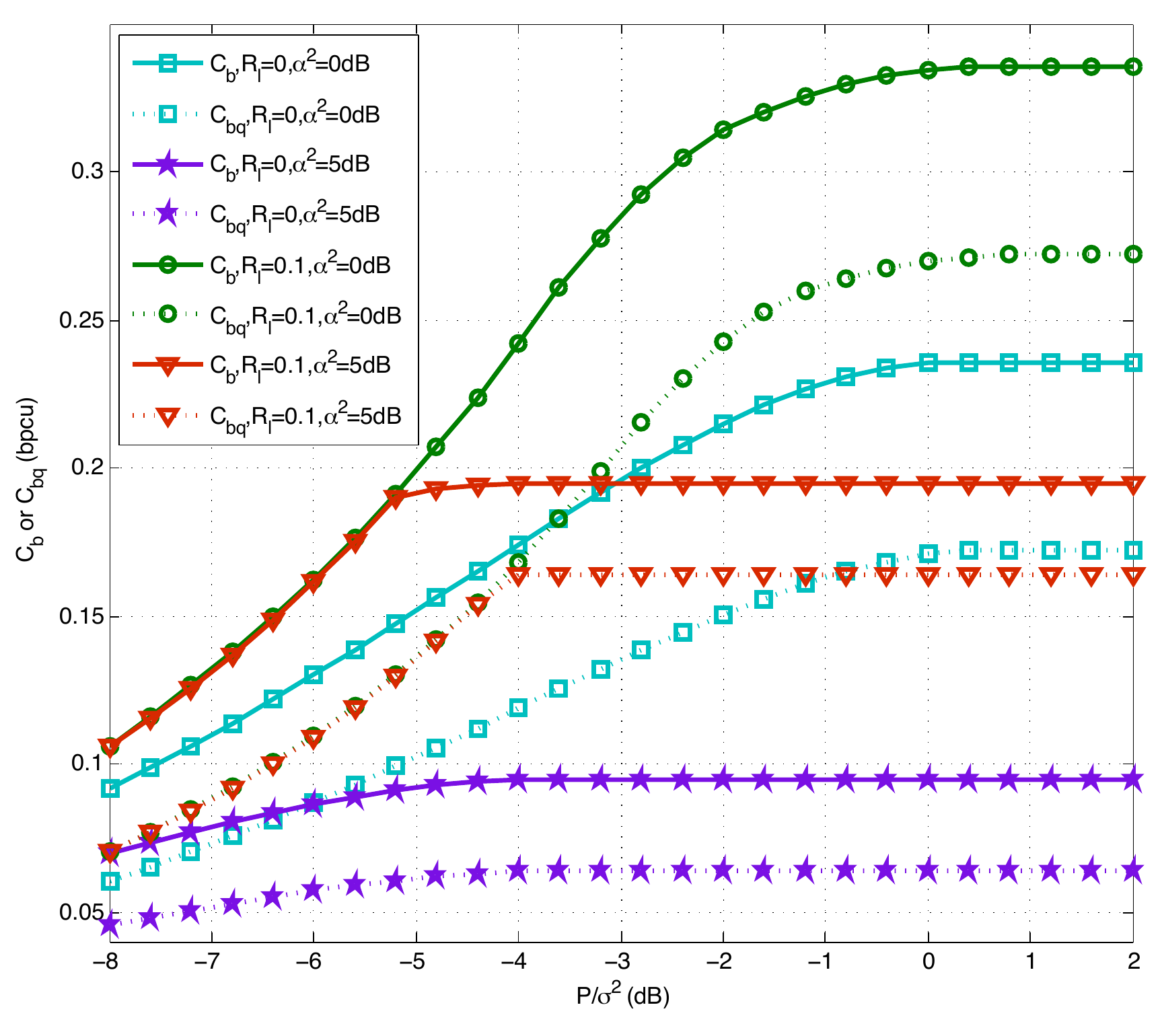}
\end{center}
\vspace*{-20pt}
\caption{Comparison between the relaxed key capacities $C_b$ and
  $C_{bq}$ for different values of maximum allowable leakage rate
  $R_l$ over the BSPK-constrained Gaussian wiretap channel. For
  $C_{bq}$, symbol-by-symbol hard-decision quantization is imposed at
  the destination.}
\label{fig:ck_p} 
\end{figure}

\begin{figure}[H]
\centering
\begin{center}
\includegraphics[width=0.57\textwidth]{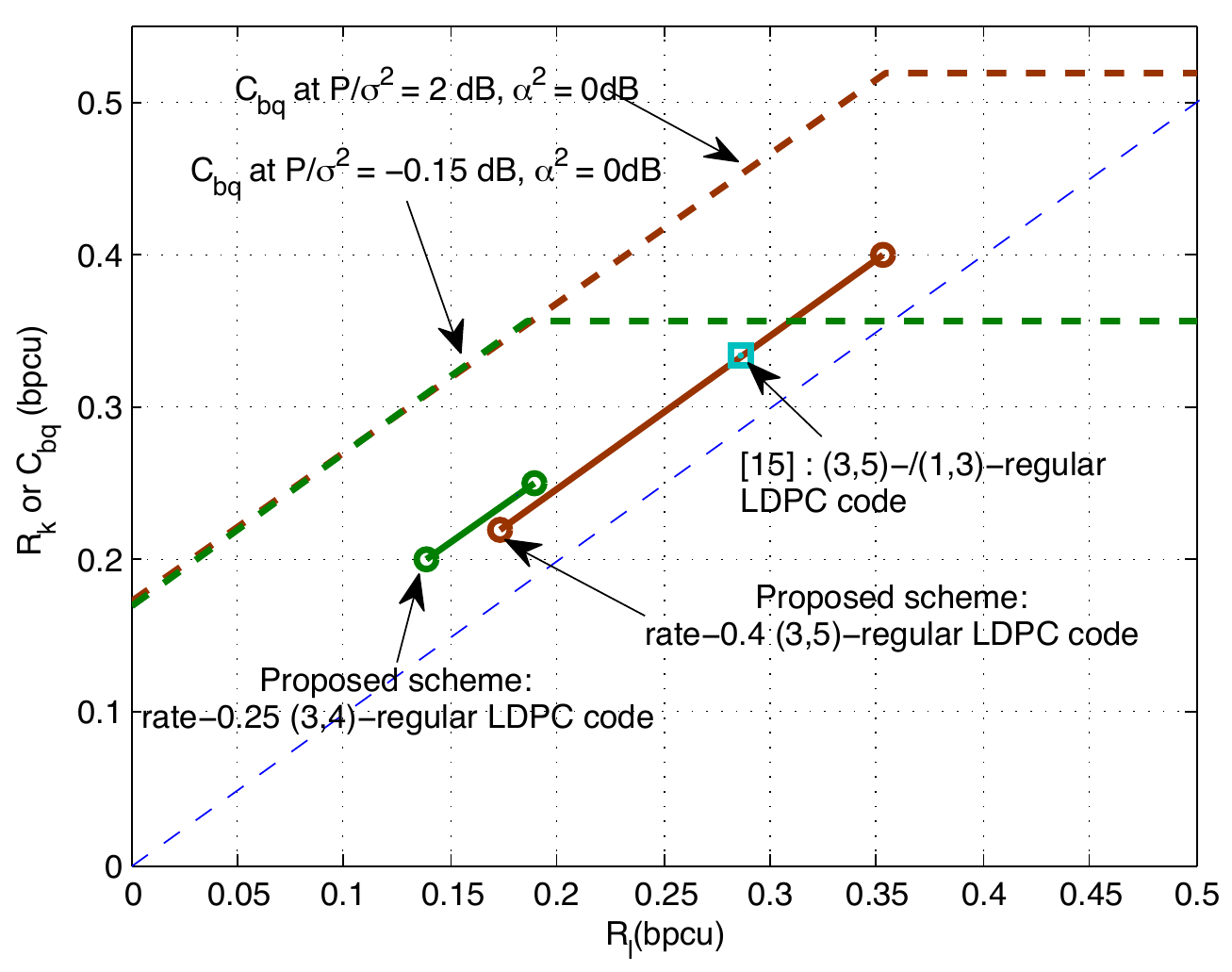}
\end{center}
\vspace*{-20pt}
\caption{Plot of the $(R_k,R_l)$-trajectories achieved by the proposed
  secret sharing scheme employing secret sharing regular LDPC codes
  $(\mathcal{C},\mathcal{W})$ with block length of $10^{5}$. Two cases
  are shown in the figure. The green curve corresponds to the case of
  $P/\sigma^2=-0.15$~dB, $\alpha^2=0$~dB, and $\mathcal{C}$ is a
  rate-$0.25$ $(3,4)$-regular LDPC code. The brown curve corresponds to
  the case of $P/\sigma^2=2$~dB, $\alpha^2=0$~dB, and $\mathcal{C}$ is
  a rate-$0.4$ $(3,5)$-regular LDPC code. For comparison, the
  corresponding boundary of the $(C_{bq},R_l)$ region for each case is
  also included in the figure. For the second case, the $(R_k,R_l)$
  rate pair achieved by the scheme proposed in~\cite{Muramatsu2006} is
  denoted by the square symbol. The code used in that scheme is
  obtained by concatenating the $(3,5)$-regular LDPC parity-check
  matrix and another $(1,3)$-regular LDPC parity-check matrix.}
\label{fig:ck_rl_regular} 
\end{figure}

\begin{figure}
\centering
\begin{center}
\includegraphics[width=0.57\textwidth]{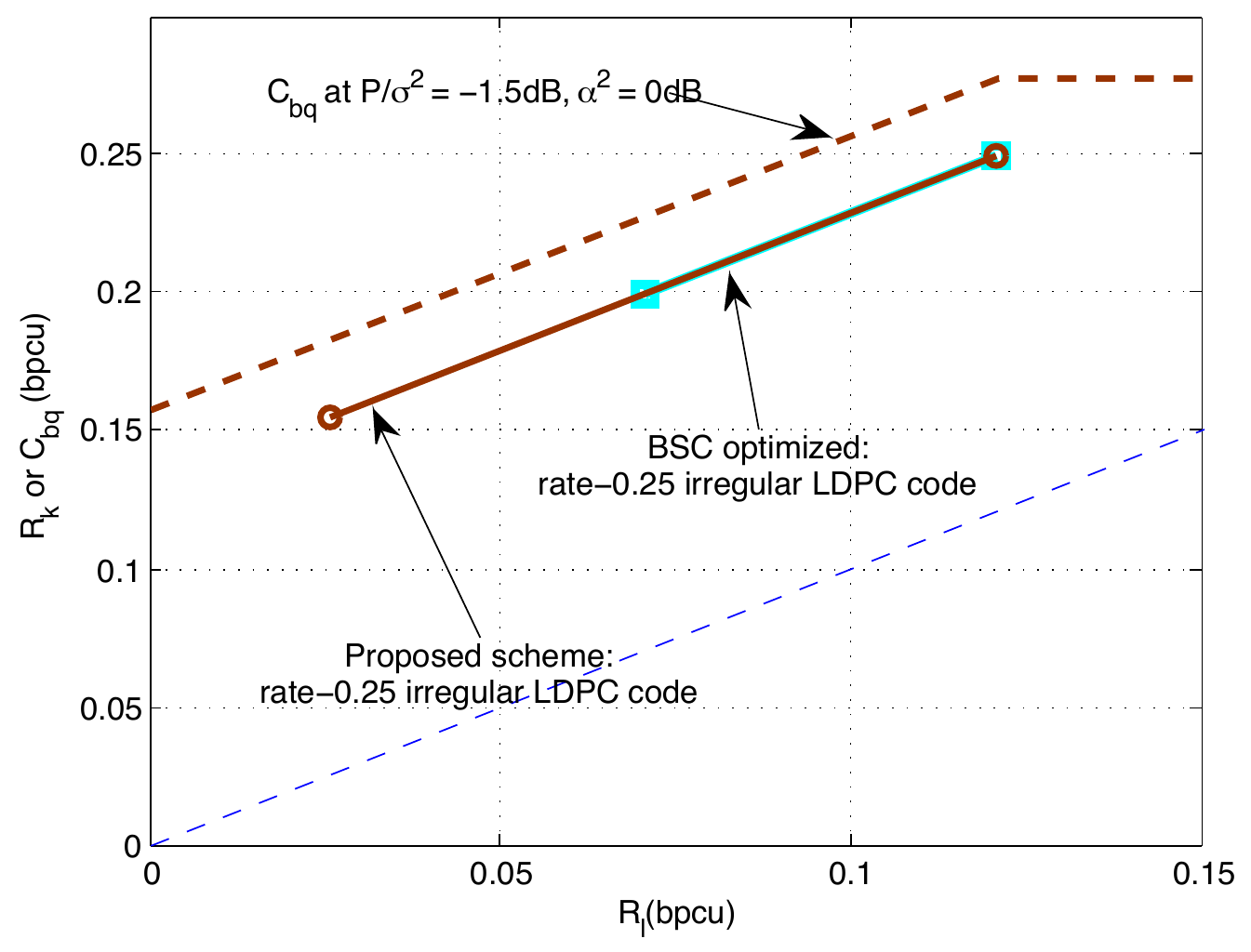}
\end{center}
\vspace*{-20pt}
\caption{Plot (with circle markers) of the $(R_k,R_l)$-trajectory
  achieved by the proposed secret sharing scheme employing the
  rate-$0.25$ secret sharing irregular LDPC code obtained from the
  code search process described in~\autoref{subsec:irregldpc}. The
  block length is set to $10^{5}$.  The channel parameter setting of
  $P/\sigma^2=-1.5$~dB and $\alpha^2=0$~dB is assumed.  The boundary
  of the $(C_{bq},R_l)$ region for this set of channel parameters is
  included in the figure. The $(R_k,R_l)$-trajectory achieved by the
  proposed secret sharing scheme employing a standard rate-$0.25$
  BSC-optimized irregular LDPC code instead is also plotted (with
  square markers) for comparison.}
\label{fig:ck_rl_irregular_0dB} %\vspace*{-20pt}
\end{figure}

\begin{figure}
\centering
\begin{center}
\includegraphics[width=0.57\textwidth]{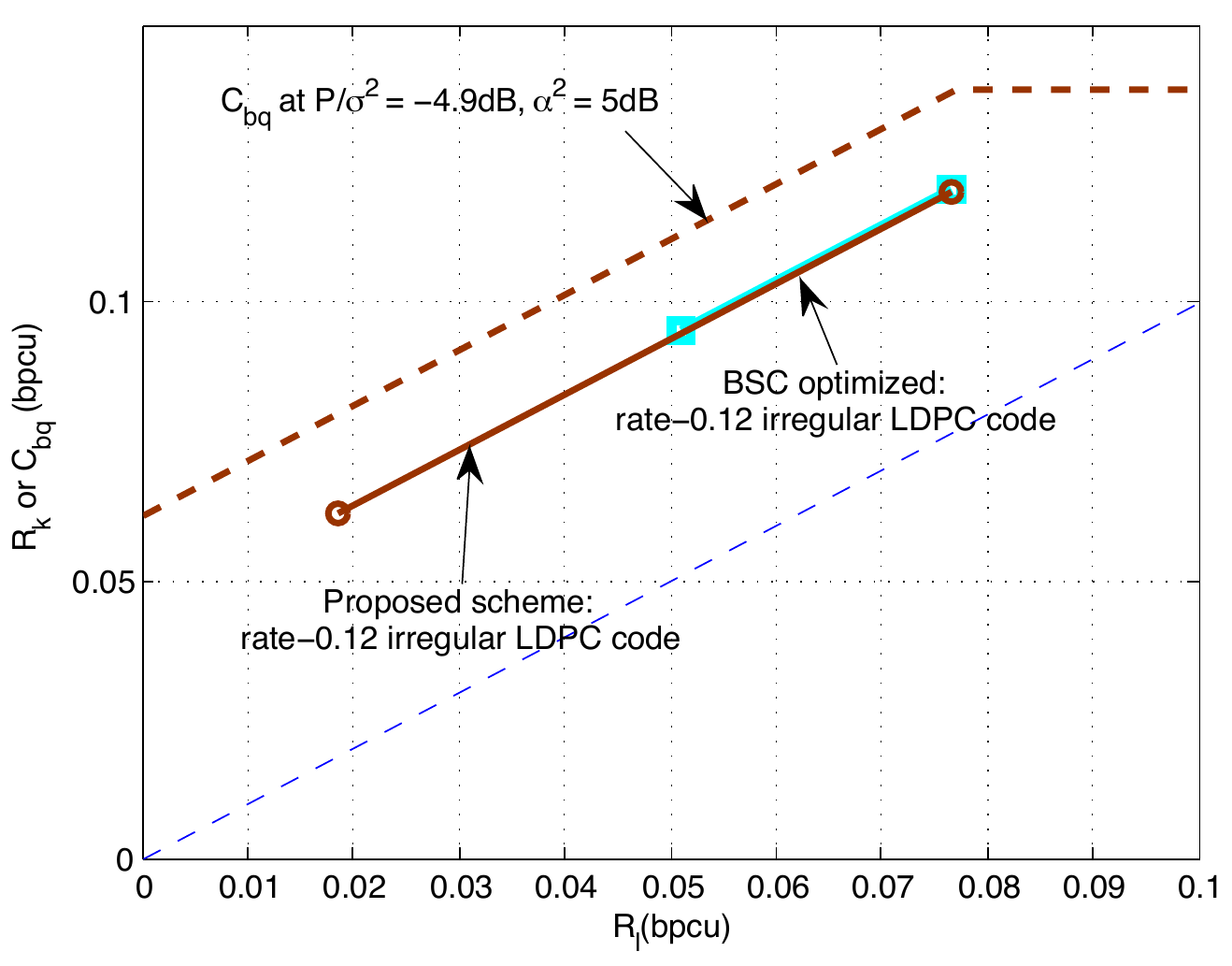}
\end{center}
\vspace*{-20pt}
\caption{Plot (with circle markers) of the $(R_k,R_l)$-trajectory
  achieved by the proposed secret sharing scheme employing the
  rate-$0.12$ secret sharing irregular LDPC code obtained from the
  code search process described in~\autoref{subsec:irregldpc}. The
  block length is set to $10^{5}$.  The channel parameter setting of
  $P/\sigma^2=-4.9$~dB and $\alpha^2=5$~dB is assumed.  The boundary
  of the $(C_{bq},R_l)$ region for this set of channel parameters is
  included in the figure. The $(R_k,R_l)$-trajectory achieved by the
  proposed secret sharing scheme employing a standard rate-$0.12$
  BSC-optimized irregular LDPC code instead is also plotted (with
  square markers) for comparison.}
\label{fig:ck_rl_irregular_5dB} %\vspace*{-20pt}
\end{figure}

\end{document}